\documentclass[aps,prx,twocolumn,showpacs,superscriptaddress,nofootinbib]{revtex4-1}
\usepackage{amsthm,amsmath,amssymb,amsfonts,amssymb}
\usepackage{hyperref}
\usepackage{color}

\usepackage{bbm,dsfont}%fonts
\usepackage{graphicx,float}
\usepackage{epstopdf}

\newtheorem{proposition}{Proposition}
\newcommand{\R}{\mathbb R} %real

\newcommand{\C}{\mathbb C} %complex
\newcommand{\mo}[1]{\left| #1 \right|} %modulus
\newcommand{\abs}{\mo} %modulus

\newcommand{\hi}{\mathcal{H}} %Hilbert space H
\newcommand{\trh}{\mathcal{T}(\mathcal{H}_S)} %trace class operators H
\newcommand{\sh}{\mathcal{S}(\mathcal{H}_S)} %states in H_system

\newcommand{\se}{\mathcal{S}(\mathcal{H}_E)} %states in enviroment
\newcommand{\ket}[1]{|#1\rangle} %ket
\newcommand{\bra}[1]{\langle#1|} %bra
\newcommand{\tr}[1]{\textrm{tr}\left[#1\right]} %trace
\newcommand{\id}{\mathbbm{1}} %identity operator
\newcommand{\supp}{\textrm{supp}} %support
\newcommand{\mc}[1]{\mathcal{#1}} %mathcal for channels
\newcommand{\expo}[1]{\textrm{exp}[#1]} %exponential

\newcommand{\Co}{\mathsf{C}}%generic observable
\newcommand{\Ho}{\mathsf{H}}%generic hamiltonian
\begin{document}

\title[]{Fidelity of dynamical maps}

%\title[]{Interaction independent quantum probing}

%\author{Mikko Tukiainen}
%\address{\textbf{Mikko Tukiainen}; Turku Centre for Quantum Physics, Department of Physics and Astronomy, University of Turku, Finland}
%\email{mikko.tukiainen@utu.fi}

%\author{Lyyra Henry}
%\address{\textbf{Henri Lyyra}; Turku Centre for Quantum Physics, Department of Physics and Astronomy, University of Turku, Finland}

%\author{Gniewomir Sarbicki}
%\address{\textbf{Gniewomir Sarbicki}; Institute of Physics, Nicolaus Copernicus University, Grudzi\c{a}dzka 5/7, 87--100 Toru\'n, Poland}

%\author{Sabrina Maniscalco}
%\address{\textbf{Sabrina Maniscalco}; Turku Centre for Quantum Physics, Department of Physics and Astronomy, University of Turku, Finland}

\author{Mikko Tukiainen}
\affiliation{Turku Centre for Quantum Physics, Department of Physics and Astronomy, University of Turku, FI-20014 Turun yliopisto, Finland}
\author{Henri Lyyra}
\affiliation{Turku Centre for Quantum Physics, Department of Physics and Astronomy, University of Turku, FI-20014 Turun yliopisto, Finland}
\author{Gniewomir Sarbicki}
\affiliation{Institute of Physics, Nicolaus Copernicus University, Grudzi\c{a}dzka 5/7, 87--100 Toru\'n, Poland}
\author{Sabrina Maniscalco}
\affiliation{Turku Centre for Quantum Physics, Department of Physics and Astronomy, University of Turku, FI-20014 Turun yliopisto, Finland}
\affiliation{Centre for Quantum Engineering, Department of Applied Physics,  School of Science, Aalto University, P.O. Box 11000, FIN-00076 Aalto, Finland}

\begin{abstract}
We introduce the concept of fidelity for dynamical maps in an open quantum system scenario. We derive an inequality linking this quantity to the distinguishability of the inducing environmental states. Our inequality imposes constraints on the allowed set of dynamical maps arising from the microscopic description of system plus environment. Remarkably, the inequality involves only the states of the environment and the dynamical map of the open system and, therefore, does not rely on the knowledge of either the microscopic interaction Hamiltonian or the environmental Hamiltonian characteristic parameters. 
We demonstrate the power of our result by applying it to two different scenarios: quantum programming and quantum probing. In the first case we use it to derive bounds on the dimension of the processor for approximate programming of unitaries. In the second case we present an intriguing proof-of-principle demonstration of the ability to extract information on the environment via a quantum probe without any {\it a priori} assumption on the form of the system-environment coupling Hamiltonian.

%\begin{description}
% \item[PACS numbers]
%PACS NUMBERS
% \end{description}
\end{abstract}

\maketitle
%%%%%%%%%%%%%%%%%%%%%%%%%%
\section{Introduction}\label{sec:intro}
%%%%%%%%%%%%%%%%%%%%%%%%%%
Quantum systems are extremely sensitive to noise arising from the interaction with their surroundings. This feature is at the very heart of the transition between the quantum and the classical description of the world and, at the same time, imposes limitations on the efficiency of quantum devices for quantum technologies. For this reason, a number of theoretical and experimental approaches focus on the modelization, characterization and reduction of noise induced by the environment, both by increasing the isolation of the system and by suitable engineering of the environment. 

A key quantity in the dynamical description of open quantum systems is the dynamical map, a time-parametrized family of quantum channels \cite{Breuer}. Formally, the exact description of the dynamical map can be obtained starting from a microscopic Hamiltonian model of the system, the environment and their interaction. In many practical circumstances, however, one does not have accurate knowledge of the microscopic details of either the interaction or the environmental Hamiltonian. In addition, it is well known that different environments and couplings can lead to the same dynamical map. It is therefore clear that, both from an experimental and from a theoretical perspective, it is crucial to develop approaches able to i) discriminate between different dynamical maps; ii) identify the restrictions imposed on the form of the dynamical map from the specific choices of the environmental state, Hamiltonian and coupling, and iii) develop techniques to probe the environment without ideally any {\it a priori} assumption on its microscopic properties and on the way it is coupled to the system.

In this paper we address these issues by introducing a new family of channel fidelities which generalize the familiar (Uhlmann) fidelity of quantum states. Endowed with these quantifiers, we prove a powerful inequality between the fidelities of the two environmental states inducing the dynamics and the corresponding dynamical maps. More precisely, the introduced fidelity quantifies the difference between two dynamical maps by measuring how distinguishable they render a pair of initially identical (or more generally non-orthogonal) states, which are initially uncorrelated from their dynamics-inducing environments. This addresses point i). The derived inequality, on the other hand, allows us to answer the points ii) and iii) since, contrarily to the probing approaches existing in the literature \cite{recati2005, bruderer2006, johnson2011, dorner2013, mcendoo2013, haikka2013, haikka2014, Anil151, Anil152}, it does not rely on the knowledge of either the system-environment interaction Hamiltonian or the characteristic parameters of the environmental Hamiltonian. We prove the usefulness of our inequality by applying it to two different scenarios both highly relevant for the development of quantum technologies, namely quantum programming and quantum probing.

A programmable quantum processor, or gate array, is a device that implements different quantum channels on one system (the data register) depending on the state of another quantum system (the program register). The third element of this device is a fixed array of gates acting on both the data and the program register. %In other words, the data register contains the quantum system on which the map is applied and the map is induced via the gate array according to the instructions set in the quantum state of the program register. 
The virtue of this arrangement is its versatility: one can realize various maps simply by controlling the state of program register. The alteration of the programming state to perform the desired operation is called quantum programming.
Programmable quantum processors were first considered by Nielsen and Chuang in the late 90's  \cite{NielsenChuang}.  Remarkably, these authors showed that  a deterministic universal programmable processor cannot be realised. The reason for this impossibility is that any two states of the program register implementing a pair of inequivalent unitary transformations need to be orthogonal \cite{NielsenChuang,Hillery02, Hillery06,HeiTuk2015}. Therefore, the programming resources, bounded by the dimension of the program register, are insufficient to implement all unitaries, comprising an uncountable set. 

The no-programming theorem above holds only if the programming is deterministic. Indeed, it was shown already in Ref.\,\cite{NielsenChuang} that it is possible to successfully implement any one-qubit unitary operation with $1/4$ probability, thus proving the existence of probabilistic universal quantum processors. This result was generalized in Refs.\,\cite{VidalarX, Vidal02}, where it was proven that the probability of success is $ p=1-\varepsilon$ and the error $\varepsilon$ can be made arbitrarily small; the results were later extend for qudits in Refs.\,\cite{Hillery01, Hillery04}. Probabilistic programming, however, tells little about programming states of general quantum transformations, described by completely positive and trace preserving (CPTP) maps. This is, however, important as all realistic implementations of quantum devices are subjected to environmental noise. In this more realistic scenario, known as approximate quantum programming, maps are near to, but not exactly, unitary. 

Approximate quantum programming has been studied in Refs.\,\cite{VidalarX, Hillery02, Hillery06}. In particular, in Refs.\,\cite{Hillery02, Hillery06} relations between the overlap of pure programming states and the programmed channels are presented. %However, these relations are processor dependent and hence the results do not allow one to tackle efficiently problems such as finding the minimal quantum circuit design able to implement the desired transformations. Moreover, as we argued above, the microscopic description of the environmental noise affecting the processor may not be known and the relations fails to describe the mutual dependencies of the programming states. Our inequality, on the other hand, sets a novel upper bound for the programming states in terms of the programmed channels and is independent of the choice of the processor. In addition, the inequality holds for general (mixed) programming states. In other words, our bound is universally valid and enables us to handle the optimality questions in full generality. 
Our inequality, on the other hand, sets a novel upper bound for the programming states in terms of the programmed channels and is independent of the choice of the processor. This independence is important since, as we argued above, the microscopic description of the environmental noise affecting the processor may not be known. In addition, unlike in the previously proposed relations, our inequality holds for general (mixed) programming states. In other words, our bound is universally valid and enables us to tackle efficiently problems such as finding size bounds for quantum circuit designs able to implement the desired transformations in full generality. 

The second application is in the framework of quantum probing. One particularly renowned quantum feature is summarized by the statement {\it no information without disturbance}. Namely, any measurement producing meaningful information generally transforms the state of the measured system. An observer is thus faced with a dilemma: how to ascertain and assign values to properties of the measured system since they may be, and typically are, altered due to the measurement? One way of obtaining some information on the system of interest without predisposing it to a direct measurement is to let it temporarily interact with a smaller ancillary quantum system and then perform a measurement on it: this technique is known as quantum probing. The system of interest in this framework acts as the environment of the quantum probe, and the induced probe dynamics carries information on the properties of the environment. If the coupling between the probe and the system of interest is sufficiently weak, measurements on the probe only lead to small perturbations on the system. In addition, measuring the full probe dynamics ideally allows one to extract even complete information on the system while leading to only minimal disturbance \cite{recati2005, bruderer2006, johnson2011, dorner2013, mcendoo2013, haikka2013, haikka2014}. 

Typical quantum probing strategies assume the knowledge of the microscopic system-probe interaction as well as the system Hamiltonian. The novelty of our method is to use the quantum programming perspective in order to go beyond these existing probing approaches. Namely, by comparing the dynamics induced by an unknown environmental state with the one of a calibration state, some properties on the measured system can be extracted with none or minimal {\it a priori} assumptions.Indeed, we provide examples in which such a probing protocol can be used for parameter estimation without knowing anything about the system-probe coupling or the probe Hamiltonian.

The paper is structured as follows. In Sect.\,\ref{sec:prel} we introduce the mathematical methods and the notation used throughout this work and provide the main result of this paper: the inequality between the channel-fidelity of two state transformations and the fidelity of the corresponding two inducing states. In  Sect.\,\ref{sec:applic} we present four applications that capture the power of  this result from different perspectives. Finally, in Sect.\,\ref{sec:conc} we summarize our results and present conclusions.

%%%%%%%%%%%%%%%%%%%%%%%%%%%%%%%%%%%%%%%%%%%%%
\section{Inequality between inducing states and open system dynamics}\label{sec:prel}
%%%%%%%%%%%%%%%%%%%%%%%%%%%%%%%%%%%%%%%%%%%%%
%%%        %%%%%%%%%%%%%%         %%%%%%%%%%%%       %%%%%%%%
We denote a complex separable Hilbert space by $\hi$ %, with either finite or countably infinite dimension,
 and the trace class operators on $\hi$ by $\mc{T(H)}$. %$\trh$. 
A \emph{quantum state} is represented as a positive operator $\varrho \in \mc{T(H)}$ with $\tr{\varrho}=1$ and the set of quantum states of $\hi$ is denoted by $\mc{S(H)}$. We indicate with $\supp(\varrho)$ the support of the state $\varrho$ and say that $\varrho_1$ and $\varrho_2$ are orthogonal, i.e. $\varrho_1 \perp \varrho_2$, when the sets $\supp(\varrho_1)$ and $\supp(\varrho_2)$ are orthogonal. The extremal elements of $\mc{S(H)}$ are called \emph{pure} and such element can be written as $\varrho = |\varphi\rangle\langle\varphi|$ for some unit vector $| \varphi \rangle \in \hi$. The transformations of quantum states %, {\it quantum channels},
 are represented by linear mappings $\mc E: \mc{T(H)} \rightarrow \mc{T(H')}$ that are completely positive and trace-preserving: such transformations are called {\it quantum channels}. In particular, a channel $\mc U$ is called a {\it unitary channel} whenever it is of the form $\mc U(\varrho) = U \varrho U^{\dagger}$ for some unitary operator $U$ on $\hi$ and for all $\varrho \in \mc{S(H)}$. 
 %%%        %%%%%%%%%%%%%%         %%%%%%%%%%%%       %%%%%%%%

%%%%%%%%
The dynamics of a closed quantum system with initial state $\varrho$ is described by a $t$-parametrized group of unitary operators $U^{(t)}$ as $\mc U^{(t)}(\varrho)= U^{(t)}\varrho U^{(-t)}$, $t\geq0$. In reality every physically realizable quantum system $S$, with Hilbert space $\hi_S$, interacts with some environment, $E$, with Hilbert space $\hi_E$, such that the pair $S + E$ can be considered a closed system, represented by $\hi_S\otimes\hi_E$. One commonly assumes that the system and the environment are initially uncorrelated, that is $\varrho_{S + E} = \varrho\otimes\xi$. The reduced dynamics of the system state can be calculated from the total system dynamics $\mc U^{(t)}(\varrho\otimes\xi)$ as a partial trace over $\hi_E$:
\begin{eqnarray}\label{eq:redchannel}
\mc E^{(t)} (\varrho) = \text{tr}_E[\mc U^{(t)}(\varrho\otimes\xi)]\,.
\end{eqnarray}
In particular, we say that the dynamics $\mc E^{(t)}(\varrho) $ is {\it induced} by the environmental state $\xi$. 
%The environment-induced dynamics usually leads to loss of quantum properties of the system, such as superposition and entanglement. Therefore it is considered the main obstacle for large-scale practical applications of quantum technologies, including quantum computers and quantum communication schemes.
%Recently, however, the effects of environment-induced dynamics have been studied from the point of view of exploiting the environmental effects, instead of seeing them as disadvantages \blue{[references]}. For example, the reservoir engineering techniques make it possible to induce information back-flow from environment to the system by manipulating the classical parameters defining the environment, such as temperature, frequency distribution, or strength of some external field \blue{[references]}.
%
% Since measurements in general disturb quantum states, the goal is to obtain some information of a system of interest without predisposing it to a direct measurement. Instead, we let it interact with some ancillary system and then we perform a measurement on the ancilla. This way there may be only a small change in the system state, but the measurement on the ancilla lets us extract some information about it. One way to measure the environmental parameters is this so-called {\it quantum probing} strategy. In this scenario the environment is our system of interest and measurements on the open system allow us to gain information about the environment parameters.

The amount of information ascribed in the state $\varrho$ is measured in terms of entropy. In our investigation, we are focusing on a class of $\alpha$-R\'{e}nyi entropies $S_\alpha(\varrho) \doteq \frac{1}{1-\alpha} \ln \tr{\varrho^\alpha}$, $\alpha\in(0,1)\cup(1,\infty)$. The $\alpha$-R\'{e}nyi entropies generate a one-parameter family of divergences, that from a certain axiomatic point of view for classical random variables $X_1$ and $X_2$ takes the unique form
\begin{eqnarray}\label{eq:defrenyicl}
D_\alpha(X_1 || X_2) = \frac{1}{\alpha-1} \sum_{i=1}^n \ln\big[ p_1(x_i) ^\alpha p_2(x_i)^{1-\alpha} \big] \, ,
\end{eqnarray}
where the two probability measures $p_1$ and $p_2$ measure the probabilities for outcome $x_i$, $i=1,\ldots, n$, occurring in $X_1$ and $X_2$, respectively \cite{Renyi}. 
In a quantum scenario, however, the uniqueness of the formulation of the divergences arising from the same set axioms is not guaranteed, and two divergent extensions of Eq.\,\eqref{eq:defrenyicl} to the quantum setting have been proposed in the literature 
\begin{eqnarray}\label{eq:defrenyi0}
\widetilde{S}_\alpha(\varrho_1 || \varrho_2)\doteq \left\{ \begin{array}{ll}
\frac{1}{\alpha -1} \ln\left[ \tr{ \varrho_1^\alpha \varrho_2^{1-\alpha}} \right],& \varrho_1 \not\perp \varrho_2 \\
\infty, & \text{otherwise},
\end{array} \right. 
\qquad
\end{eqnarray}
and
\begin{eqnarray}\label{eq:defrenyi}
&&S_\alpha \left( \varrho_1 || \varrho_2 \right) \doteq \nonumber \\
& & \left\{ \begin{array}{ll}
\frac{1}{\alpha-1} \ln \left[ \tr{\left(\varrho_2^{\frac{1-\alpha}{2\alpha}} \varrho_1 \, \varrho_2^{\frac{1-\alpha}{2\alpha}} \right)^\alpha} \right], & \text{when} \ \varrho_1\not\perp \varrho_2  \\
 \infty, & \text{otherwise,}
\end{array} \right. 
\end{eqnarray}
defined for $\alpha \in (0,1)$ and a pair of quantum states $\varrho_1, \varrho_2 \in \sh$ \cite{Muller2013, Wilde2014}. In this study, we will be focusing for the latter definition $S_\alpha$ and call it the {\it quantum $\alpha$-R\'{e}nyi divergence}. It is, however, worth mentioning that all the results we will present in this Article could be also formulated in the context of $\widetilde{S}_\alpha$ and furthermore, since the $S_{\alpha}$ and $\widetilde{S}_\alpha$ coincide for commuting states $\big[\varrho_1, \varrho_2 \big] =0$ \cite{Muller2013}, most of these results are in fact equivalent for the two definitions.

Importantly, many of the most commonly used quantum (relative) entropies, such as the common (relative) entropy, the min- and the max-(relative)-entropies, can be derived from the quantum $\alpha$-R\'{e}nyi divergence as special cases; for the properties of $\alpha$-R\'{e}nyi divergence we refer the reader to \cite{Beigi2013, Datta2014, Datta2016, Carlen2016} and references therein. For our purposes, we recall here the following features
\begin{itemize}
\item[(S1)] $S_\alpha \big(\varrho_1||\varrho_2\big)\geq 0$, where the equality holds if and only if $ \varrho_1 = \varrho_2$,
\item[(S2)] $S_\alpha\big(\varrho_1 \otimes \xi_1 || \varrho_2 \otimes \xi_2\big) = S_\alpha\big(\varrho_1 || \varrho_2\big) + S_\alpha\big(\xi_1 || \xi_2\big)$,
\item[(S3)] $S_\alpha\big(\mc U(\varrho_1) || \mc U(\varrho_2)\big) = S_\alpha\big(\varrho_1 || \varrho_2\big)$
\end{itemize}
for all  $\alpha\in(0,1)$ and $\varrho_1,\varrho_2 \in \sh$,  $\xi_1,\xi_2 \in \se$ and for all unitary channels $\mc U$. Notice, that the order of inputs is important, since in general $S_{\alpha} \big(\varrho_1, \varrho_2\big) \neq S_{\alpha} \big(\varrho_2,\varrho_1\big)$. Additionally, for all $\alpha\in[1/2,1)$, the $\alpha$-R\'{e}nyi divergence satisfies the data processing inequality
\begin{itemize}
\item[(S4)] $S_\alpha\big(\mc E(\varrho_1) || \mc E(\varrho_2)\big) \leq S_\alpha\big(\varrho_1 || \varrho_2\big)$, 
\end{itemize}
for all  $\varrho_1, \varrho_2 \in \sh$ and for arbitrary quantum channels $\mc E:\trh \rightarrow \mc T(\hi_S')$.
% $\hi_S'$ being arbitrary Hilbert space \cite{Beigi2013}. 

From Eq.\,\eqref{eq:redchannel} it is clear that altering the inducing state $\xi\in\se$, while keeping the coupling interaction fixed, may lead to different channels. A question then arises: How do the two different inducing states and the corresponding induced channels relate to each other? The following proposition provides insight into this question.
\begin{proposition}
Suppose that two (different) channels $\mc E_1$ and $\mc E_2:\trh \rightarrow \trh$ are induced by states $\xi_1$ and $\xi_2\in \se$, respectively, from some fixed (unitary) coupling $\mc U:\mc T(\hi_S\otimes\hi_E)\rightarrow\mc T(\hi_S\otimes\hi_E)$; see Eq.\,\eqref{eq:redchannel}. Then for all $\alpha\in[1/2,1)$ and $\varrho_1,\varrho_2 \in \sh$
\begin{eqnarray}\label{eq:mainineq1}
S_\alpha \big(\mc E_1(\varrho_1)||\mc E_2(\varrho_2)\big) - S_\alpha\big(\varrho_1||\varrho_2\big) \leq  S_\alpha\big(\xi_1||\xi_2\big).
\end{eqnarray} 
\end{proposition}
\begin{proof}
The claim follows from the properties (S1)-(S4), namely
  \begin{eqnarray}
&&   S_\alpha \big( \varrho_1 || \varrho_2 \big) + S_\alpha \big( \xi_1 || \xi_2 \big) \nonumber \\
&=&  S_\alpha \big( \mc U(\varrho_1\otimes\xi_1) || \mc U(\varrho_2\otimes\xi_2) \big) \nonumber \\
   &\geq &    S_\alpha \big( \text{tr}_{E} \left[ \mc U(\varrho_1\otimes\xi_1) \right] || \text{tr}_{E} \left[ \mc U(\varrho_2\otimes\xi_2)\right] \big) \nonumber \\
   &=&   S_\alpha \big(\mc E_1(\varrho_1)||\mc E_2(\varrho_2)\big) \, .
  \end{eqnarray}
\end{proof}
We will use a short hand notation 
\begin{eqnarray}
&&F_{\alpha} \big(\varrho_1, \varrho_2\big) \doteq \tr{\left(\varrho_2^{\frac{1-\alpha}{2\alpha}} \varrho_1 \, \varrho_2^{\frac{1-\alpha}{2\alpha}} \right)^\alpha},
\end{eqnarray} since it will often appear in the following, and call the quantity $F_{\alpha}$ the {\it $\alpha$-fidelity of states} \footnote{In one of the founding papers  \cite{Wilde2014}, this quantity goes by the name ``sandwiched quasi-relative entropy''.}. In the particular case $\alpha=1/2$, $F_{1/2}$ corresponds to the usual (Uhlmann) fidelity of states. Using this notation and the monotonicity of the logarithm function the Ineq.\,(\ref{eq:mainineq1}) implies
\begin{eqnarray}\label{eq:mainineq2}
F_{\alpha} \big(\varrho_1, \varrho_2\big) F_{\alpha} \big(\xi_1, \xi_2\big) \leq F_{\alpha} \big(\mc E_1(\varrho_1), \mc E_2(\varrho_2)\big),
\end{eqnarray} 
for $\alpha\in[1/2,1)$, which can be seen as a generalization of the data processing inequality $F_{\alpha} \big(\varrho_1, \varrho_2\big) \leq F_{\alpha} \big(\mc E(\varrho_1), \mc E(\varrho_2)\big)$, for $\alpha\in[1/2,1)$.  Motivated by this inequality, for $\alpha\in(0,1)$, we define
\begin{eqnarray}\label{eq:chanafid} \mc F_{\alpha}(\mc E_1, \mc E_2) \doteq \inf_{\varrho_1,\varrho_2\in\sh} \frac{ F_{\alpha} \big(\mc E_1(\varrho_1), \mc E_2(\varrho_2)\big) }{ F_{\alpha} \big(\varrho_1, \varrho_2 \big) },
\end{eqnarray} and call it the {\it $\alpha$-fidelity of channels} $\mc E_1$ and $\mc E_2$. We justify the terminology, since $\mc F_{\alpha}$ shares many of the basic properties of fidelity-measures as we shall verify next.
\begin{proposition}
For all $\alpha\in(0,1)$ and channels $\mc E_i:\mc T(\hi_S)\to\mc T(\hi_S)$, $i=1,2$, $\alpha$-fidelity of channels $\mc F_{\alpha}(\mc E_1, \mc E_2)$ has the following properties:
\begin{itemize}
%\item[(F0)] $\mc F_{\alpha}(\mc E_1, \mc E_2)=\mc F_{1-\alpha}(\mc E_2, \mc E_1)$,
\item[(F1)] $\mc F_{\alpha}(\mc E_1, \mc E_2)\in [0,1]$,
\item[(F2)] $\mc F_{\alpha}(\mc E_1, \mc E_2)=\mc F_{\alpha}(\mc U \circ \mc E_1 \circ \mc V , \mc U\circ \mc E_2 \circ \mc V)$, for all unitary channels $\mc U$ and $\mc V$.
\end{itemize}
Additionally, for $\alpha\in[1/2,1)$,
\begin{itemize}
\item[(F3)] $\mc F_{\alpha}(\mc E_1, \mc E_2)=1$ iff $\mc E_1=\mc E_2$,
\item[(F4)] $\mc F_{\alpha}(\mc E_1, \mc E_2) \leq \mc F_{\alpha}(\mc E \circ \mc E_1,\mc E \circ \mc E_2)$,  for all channels $\mc E:\trh\rightarrow\mc T(\hi_S')$,
\item[(F5)] $\mc F_{\alpha}(\mc E_1, \mc E_2) \leq \mc F_{\alpha}(\mc E_1 \circ \mc E,\mc E_2 \circ \mc E)$, for all channels $\mc E:\mc T(\hi_S')\rightarrow\trh$.
\end{itemize}
\end{proposition}
\begin{proof}
%F(0) on a wishlist.
Clearly $\mc F_{\alpha}(\mc E_1, \mc E_2)$ is non-negative by definition. Since $F_\alpha(\varrho_1,\varrho_2)\leq 1$ for all $\varrho_1, \varrho_2 \in \sh$, as verified from (S1) of $\alpha$-R\'{e}nyi divergence
, $\mc F_\alpha\big(\mc E_1, \mc E_2\big) \leq \inf_{\varrho_1, \varrho_2} F_{\alpha}(\varrho_1,\varrho_2)^{-1} =1$ proving (F1). Using the properties (S3) and (S4) of $\alpha$-R\'{e}nyi divergence, the monotonicity of the logarithm function and the fact $\sh \simeq \mc U \big( \sh \big)$ we conclude (F2) and (F4). Since $F_{\alpha} \big(\varrho_1,\varrho_2\big) \leq F_{\alpha} \big(\mc E(\varrho_1),\mc E(\varrho_2)\big)$, for all channels $\mc E$ when $\alpha\in[1/2,1)$, we confirm $\mc F_{\alpha}(\mc E, \mc E) = 1$. Assume on the other hand that $\mc F_{\alpha}(\mc E_1, \mc E_2) = 1$, but make a counter assumption $\mc E_1 \neq \mc E_2$. Then there exists a state $\varrho\in \sh$ such that $\mc E_1 (\varrho) \neq \mc E_2(\varrho)$ and therefore by (S1) $F_{\alpha} \big(\mc E_1(\varrho), \mc E_2(\varrho)\big)<1$. Since  $1=\mc F_{\alpha}(\mc E_1, \mc E_2) \leq F_{\alpha} \big(\mc E_1(\varrho), \mc E_2(\varrho)\big)<1$ this leads to a contradiction and we get (F3). Lastly we prove (F5). From the properties above we get the inequality 
\begin{align}
&\frac{F_{\alpha} \big((\mc E_1 \circ \mc E)(\varrho_1),(\mc E_2 \circ \mc E)(\varrho_2)\big) }{F_{\alpha} \big(\mc E(\varrho_1),\mc E(\varrho_2)\big)} \leq  \nonumber \\
&\frac{F_{\alpha} \big((\mc E_1 \circ \mc E)(\varrho_1),(\mc E_2 \circ \mc E)(\varrho_2)\big) }{F_{\alpha} \big(\mc \varrho_1,\varrho_2\big)}\,,
\end{align}
which together with $\mc E (\mc S(\hi_S')) \subset \mc S(\hi_S)$ implies that 
\begin{eqnarray}
\mc F_{\alpha}(\mc E_1, \mc E_2) &\leq &\inf_{\mc E(\varrho_1), \mc E(\varrho_2)} \frac{F_{\alpha} \big(\mc E_1( \mc E(\varrho_1) ),\mc E_2(\mc E (\varrho_2) )\big) }{F_{\alpha} \big(\mc E(\varrho_1),\mc E(\varrho_2)\big)} \nonumber \\
&= & \inf_{\varrho_1, \varrho_2} \frac{F_{\alpha} \big((\mc E_1 \circ \mc E)(\varrho_1),(\mc E_2 \circ \mc E)(\varrho_2)\big) }{F_{\alpha} \big(\mc E(\varrho_1),\mc E(\varrho_2)\big)} \nonumber \\
 &\leq & \inf_{\varrho_1,\varrho_2} \frac{F_{\alpha} \big((\mc E_1 \circ \mc E)(\varrho_1),(\mc E_2 \circ \mc E)(\varrho_2)\big) }{F_{\alpha} \big(\mc \varrho_1,\varrho_2\big)} \nonumber \\
 &=&  \mc F_{\alpha}(\mc E_1 \circ \mc E,\mc E_2 \circ \mc E).
\end{eqnarray}
%By using the multiplicativity $F_{\alpha} \big(\varrho_1 \otimes \xi_1, \varrho_2 \otimes \xi_2) = F_{\alpha} \big(\varrho_1, \varrho_2) F_{\alpha} \big(\xi_1, \xi_2)$ and (F4) we confirm that $F_{\alpha} \big(\mc E(\varrho_1), \mc E(\varrho_2)\big) = F_{\alpha} \big( \mc U \left(\mc E(\varrho_1)\otimes \xi_1 \right) , \mc U \left(\mc E(\varrho_2)\otimes \xi_2 \right) \big)$ for all unitaries $\mc U: \trhk \rightarrow \trhk$ and states $\xi_i\in\se$, $i=1,2$. Since the partial trace $\text{tr}_\ki: \mc T(\hi \otimes \ki) \rightarrow \trh$ is a channel, (F6) now follows from (F5).
\end{proof}

Different fidelity measures of pairs of quantum channels have been formulated including the process fidelity measuring the distinguishability of the corresponding Choi-states \cite{Raginsky2001}, the minimax fidelity \cite{Raginsky2005} with operational connection to the single-shot discrimination of channels done with the so-called process POVMs \cite{Mario2008, Mario2009} and the plethora of gate fidelities considered in \cite{Nielsen2005, NielsenInfo}. All of these channel fidelities have a common feature: they measure the distinguishability of pair of quantum channels in different operational scenarios. 
The $\alpha$-fidelity of channels we propose does not, however, admit operational connection to the above channel discrimination tasks due to a peculiar property the other channel fidelities do not share. Namely, as we will see shortly, it vanishes for any pair of different unitary channels. This unique feature makes our fidelity particularly appealing in quantum programming. 

We note that the double infima in Eq.\,\eqref{eq:chanafid} can make $\mc F_{\alpha}$ difficult to evaluate. It is, however, readily verified that $\mc F_{\alpha}\big(\mc E_1, \mc E_2\big)\leq \inf_{\varrho} F_{\alpha} \big(\mc E_1(\varrho), \mc E_2(\varrho) \big)$, where the latter quantity also sets an upper bound for the inducing states and it is much easier to handle. Furthermore, since it has been shown in \cite{Mosonyi2014} that $F_{\alpha}$ is jointly concave for $\alpha\in[1/2,1)$, that is in particular for an arbitrary $\varrho = \sum_i \lambda_i |i\rangle \langle i|$ it satisfies $F_{\alpha} \big(\mc E_1(\sum_i \lambda_i |i\rangle\langle i|), \mc E_2(\sum_i \lambda_i |i\rangle\langle i|) \big) \geq \sum_i \lambda_i F_{\alpha} \big(\mc E_1(|i\rangle\langle i|), \mc E_2(|i\rangle\langle i|) \big)$ for $\alpha\in[1/2,1)$, it can be concluded that it is enough to consider the infimum over the set of pure states. In the special case $\alpha=1/2$ the quantity $F_{\text{min}} (\mc E_1, \mc E_2) \doteq \inf_{\varrho} F_{1/2} \big(\mc E_1(\varrho), \mc E_2(\varrho) \big)$ is known as the {\it minimal gate fidelity} \cite{Nielsen2005, NielsenInfo}. We leave it as an open problem, whether also in $\mc F_{\alpha}$ it is sufficient to evaluate the infima over the set pure states or not.

Above an ordering $\mc F_{1/2} \leq F_{\text{min}}$ was pointed out, and one may wish to reveal such relations between other fidelities also, {\it e.g.} between $\alpha$-fidelity and the commonly used process fidelity defined for channels $\mc E_i: \mc T(\hi_S) \to \mc T(\hi_S), $ $i=1,2,$ via $F_{\text{proc} } (\mc E_1, \mc E_2) \doteq F_{1/2} ( \mc E_1 \otimes \mc I (\Omega), \mc E_2 \otimes \mc I (\Omega) )$, where $\mc I: \mc T(\hi_S) \to \mc T(\hi_S)$ is the identity channel and $\Omega = \frac{1}{\dim (\hi_S)} \sum_{n,m=1}^{\dim (\hi_S)} |n n\rangle \langle m m|$ is the maximally entangled state in $\hi_S \otimes \hi_S$ with respect to an orthonormal basis $\{ \ket n,\, n=1,\ldots \dim (\hi_S) \}$. Such an ordering cannot, however, be established. Namely, it can be confirmed that for the one-qubit channels $\mc E_1=\mc I$ and $\mc E_2(\varrho) = 1/3 \sum_{i=1}^3 \sigma_i \varrho \sigma_i^\dagger$, where $\sigma_i$, $i=1,2,3,$ are the Pauli matrices, the process fidelity vanishes, $F_{\text{proc} } (\mc E_1, \mc E_2) = 0$, or more generally $F_{\alpha} \big( \mc E_1 \otimes  \mc I, \mc E_2 \otimes \mc I(\Omega) \big)=0$ for all $\alpha \in (0,1)$. On the other hand, by writing $\varrho_{\vec{m}} = \frac{1}{2} ( \id + \vec m \cdot \vec \sigma),$ where $\vec m \cdot \vec \sigma = \sum_{i=1}^3 m_i \sigma_i$ and $\vec{m}\in \R^3$ has $|| \vec m || \leq 1$, we see that $\mc E_2(\varrho_{\vec{m}}) =  \frac{2}{3} \cdot (\frac 12 \id) + \frac 13 \cdot \varrho_{-\vec m}$ and therefore
\begin{eqnarray}
\mc F_{\alpha} (\mc E_1, \mc E_2) &=& \inf_{\varrho_{\vec n}, \varrho_{\vec m}} \frac{F_{\alpha} ( \varrho_{\vec n}, \frac{2}{3} \cdot (\frac 12 \id) + \frac 13 \cdot \varrho_{-\vec m})}{F_{\alpha}( \varrho_{\vec n}, \varrho_{\vec m})} \nonumber \\
&\geq & \frac 23 \inf_{\varrho_{\vec n}} F_{\alpha} ( \varrho_{\vec n}, \frac 12 \id )  + \frac 13 \inf_{\varrho_{\vec n}, \varrho_{\vec m}}  \frac {F_{\alpha}(\varrho_{\vec n}, \varrho_{-\vec m})}{F_{\alpha}(\varrho_{\vec n}, \varrho_{\vec m})} \nonumber \\
&=& \frac 23 (\frac 12)^{1-\alpha} \neq 0 \, ,
\end{eqnarray}
for $\alpha \in [1/2, 1)$, where the estimation follows from the joint concavity of the $\alpha$-fidelity \cite{Mosonyi2014}. Therefore, there exist channels $\mc E_1$ and $\mc E_2$ for which $F_{\text{proc}} (\mc E_1, \mc E_2) < \mc F_{\alpha} (\mc E_1, \mc E_2)$ for all $\alpha \in [1/2, 1)$. However, in Prop.\,\ref{prop:programming} we will show that $\mc F_{\alpha} (\mc U_1, \mc U_2) =0$ for any pair of unitary channels $\mc U_1 \neq \mc U_2$, which does not hold for $F_{\text{proc} }$ in general \cite{Raginsky2001}. In conclusion, there is no ordering between $\mc F_{\alpha}$ and $F_{\text{proc}}$. Furthermore, since the process fidelity and the minimax fidelity \cite{Raginsky2005} $F_{\text{mm}}(\mc E_1, \mc E_2) \doteq \inf_{\varrho \in \mc S(\hi_S)} \tr{|\sqrt{ \mc E_1 \otimes \mc I (\Omega)} \, \id \otimes \varrho \, \sqrt{ \mc E_2 \otimes \mc I (\Omega)}|}$ satisfy $F_{\text{mm}} \leq F_{\text{proc}}$, and since there exist unequal unitary channels $\mc U_i$, $i=1,2$, such that $F_{\text{mm}}(\mc U_1, \mc U_2) \neq 0$, it may be concluded that the channel fidelities $F_{\text{mm}}$ and $\mc F_{\alpha}$ are not related either. As a corollary, we can prove that the $\alpha$-fidelity of channels is not (sub)multiplicative for separable channels.
\begin{proposition}
$\mc F_\alpha$ does not satisfy the multiplicativity property 
\begin{eqnarray}
\mc F_{\alpha} \big( \mc E_1 , \mc E_2 \big) \,  \mc F_{\alpha} \big(  \mc C_1, \mc C_2 \big)  = \mc F_{\alpha} \big( \mc E_1 \otimes  \mc C_1, \mc E_2 \otimes \mc C_2 \big) \, ,
\end{eqnarray} for all channels $\mc E_i: \mc T(\hi_S) \to \mc T(\hi_S)$, $\mc C_i: \mc T(\hi'_S) \to \mc T(\hi'_S)$, $i=1,2$. 
\end{proposition}
\begin{proof}
It is easily seen, that $\mc F_{\alpha} ( \mc E_1 , \mc E_2) \, \mc F_{\alpha} (  \mc C_1, \mc C_2 )  \geq \mc F_{\alpha} ( \mc E_1 \otimes  \mc C_1, \mc E_2 \otimes \mc C_2 )$, since $\mc S(\hi_S) \otimes \mc S(\hi'_S) \subset \mc S (\hi_S \otimes \hi'_S)$. The other direction, however, leads to a contradiction. Namely, assume that $\mc F_{\alpha} ( \mc E_1 , \mc E_2) \mc F_{\alpha} (  \mc C_1, \mc C_2 )  \leq \mc F_{\alpha} ( \mc E_1 \otimes  \mc C_1, \mc E_2 \otimes \mc C_2 )$ holds for all channels $\mc E_i: \mc T(\hi_S) \to \mc T(\hi_S)$, $\mc C_i: \mc T(\hi'_S) \to \mc T(\hi'_S)$, $i=1,2$. Then in particular
\begin{eqnarray}
\mc F_{\alpha} \big( \mc E_1 , \mc E_2\big) &\leq & \mc F_{1/2} \big( \mc E_1 \otimes \mc I , \mc E_2 \otimes \mc I \big) \nonumber \\
&\leq & F_{\alpha}\big( \mc E_1 \otimes \mc I (\Omega), \mc E_2 \otimes \mc I(\Omega) \big) \, ,
\end{eqnarray} 
which we have noted before does not hold in general. 
\end{proof}
\noindent We observe, that similar calculations to those above can be used to verify that $F_{\text{min}}$ is not multiplicative either. 

So far we have only considered static quantum channels, however, the calculations above hold as well for dynamical maps. The emerging inequality 
\begin{eqnarray}\label{eq:mainineq}
F_{\alpha}\big(\xi_1,\xi_2\big) \leq \mc F_{\alpha}\big(\mc E_1^{(t)}, \mc E_2^{(t)}\big), \,\, \alpha\in[1/2,1), \,\, t\ge 0,
\end{eqnarray} has the following twofold interpretation. Firstly, if the initial environmental states $\xi_i$, $i=1,2$, induce dynamics $\mc E^{(t)}_i$, respectively, then $\mc F_{\alpha}\big(\mc E^{(t)}_1, \mc E^{(t)}_2\big)$ cannot decrease below $F_{\alpha}\big(\xi_1,\xi_2\big)$ for any $t\ge 0$. This sets limitations on what dynamics are compatible with the inducing environmental states. On the other hand, since $\inf_{t\ge 0} \mc F_{\alpha}\big(\mc E_1^{(t)}, \mc E_2^{(t)}\big)$ majorises $F_{\alpha}\big(\xi_1,\xi_2\big)$, it is possible to gain some information about the environment by measuring the open system dynamics, {\it i.e.} by probing. Remarkably, these two strategies work even if the environment and the interaction were not specified. The next section is devoted to examples that demonstrate these properties.

\section{Applications of the inequality}\label{sec:applic}
In the following we will present three applications of Ineq.\,\eqref{eq:mainineq2} for quantum programming and quantum probing purposes.
Figure \ref{fig:progprob} illustrates the quantum probing and quantum programming protocols.
Despite their apparent similarity, the main objectives in the two protocols are complementary: 
in quantum programming the aim is to induce different channels for the data register in a controlled way, 
whereas in quantum probing the goal is to extract information of the system of interest which induced the dynamics. 

\begin{figure}
\includegraphics[width=0.45\textwidth]{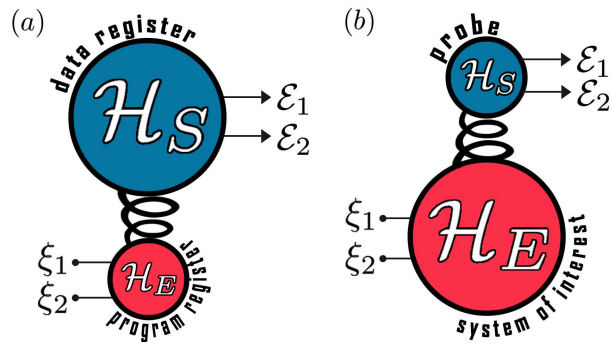}
\caption{{\bf Quantum programming and probing.} (a) In a quantum programming scheme one varies the state $\xi_i$ of the program register ($\hi_E$) coupled to the data register ($\hi_S$). 
By varying $\xi_i$ we can induce multiple different channels $\mc E_i$ on the data register. In this scenario emphasis in interest is on $\hi_S$.
(b) In a quantum probing scheme, on the other hand, the system of interest is $\hi_E$ which is coupled to the probe $\hi_S$. 
Different initial states, $\xi_i$, of the system lead to different dynamics, $\mc E_i$, of the probe.
By measuring the probe dynamics $\mc E_i$ we can obtain some information of the corresponding inducing state $\xi_i$.}
\label{fig:progprob}
\end{figure}

\subsection{Quantum programming}
The altering of the inducing state $\xi\in\se$ in Eq.\,\eqref{eq:redchannel} is an action known as {\it quantum programming} and accordingly, the pair $\langle \hi_E, \mc U \rangle$ formed by the program space $\hi_E$ and the (unitary) coupling $\mc U$ between $\hi_S$ and $\hi_E$ is called a {\it programmable processor}. 
%
%The roots of quantum programming trace back to the results of Nielsen and Chuang in the late '90s. In their seminal paper \cite{NielsenChuang} it was shown that two different unitary channels arising from the same programmable processor require orthogonal inducing states. 
%%In Ref. \cite{HeiTuk2015} it was further shown that the inducing states of a unitary and any different extremal channel are orthogonal. Here, the term {\it extremal} refers to the channels that do not admit non-trivial convex decompositions to other channels. For literature on quantum programming, we refer to \blue{CITES}.
%This result implies the impossibility of deterministic universal quantum programming, if the ancillary Hilbert space $\hi_E$ is assumed separable: it is not possible to {\it perfectly} implement all (unitary) channels with a single programmable processor by means of quantum programming due to a simple cardinality argument. This ``no-programming'' result tells, however, very little about the inducing states of quantum channels in general. To shed light on this topic is important, since often the implemented transformations differ from the unitary ones.

From the inequality
\begin{eqnarray}\label{eq:progineq}
F_{\alpha}\big(\xi_1,\xi_2\big) \leq \mc F_{\alpha}\big(\mc E_1, \mc E_2\big), \qquad \alpha \in [1/2,1),
\end{eqnarray}
we get a limit for the fidelity of the inducing states in terms of the induced channels. We recall once again that a sharp equality between the inducing states and the induced channels has been introduced in \cite{Buzek2006}. This equality depends on the programmable processor. In Ineq.\,\eqref{eq:progineq}, on the contrary, no such dependence is present.

In the case $\alpha=1/2$, Ineq.\,\eqref{eq:progineq} has a particularly nice interpretation, {\it viz.} $F_{1/2}(\xi_1,\xi_2)$ measures the distinguishability of the states $\xi_i$, $i=1,2$. Also $\mc F_{1/2}$ has a similar interpretation: $\mc F_{1/2}(\mc E_1,\mc E_2)=\inf_{\varrho_1,\varrho_2} \frac{F_{1/2}\left(\mc E_1(\varrho_1),\mc E_2(\varrho_2) \right)}{F_{1/2}(\varrho_1,\varrho_2)}$ quantifies the mutual {\it distinguishability of the channels} $\mc E_1,\mc E_2$ by measuring how distinguishable these channels can render a pair of initially non-orthogonal states $\varrho_1$ and $\varrho_2$. The inequality then naturally stems from the fact that the distinguishabilities should not increase under physical processes, such as programming. We connect Ineq.\,\eqref{eq:progineq} to
earlier known programming results \cite{NielsenChuang, Hillery02, Hillery06, HeiTuk2015} by proving that unitary channels are, in this sense, the most distinguishable.

\begin{proposition}\label{prop:programming}
Let $\mc U_1$ and $\mc U_2$ be two unitary channels. If $\mc U_1\neq \mc U_2$ then $\mc F_{\alpha}(\mc U_1, \mc U_2)=0$ for all $\alpha\in (0,1)$. In particular, the inducing states of different unitary channels must be orthogonal regardless of the choice of the processor $\langle \hi_E, \mc U\rangle$.
\end{proposition}
\begin{proof}
We first prove the proposition for $\alpha=1/2$. Using the properties of fidelity of pure states we have $F_{1/2}\big(\mc U_1 ( |\varphi_1 \rangle \langle \varphi_1| ), \mc U_2 ( |\varphi_2 \rangle \langle \varphi_2| ) \big) = | \langle \varphi_1 |  U_1^{\dagger} U_2 \varphi_2 \rangle |$. It suffices to show that for any operator $U \neq c \cdot \id_{\hi_S}$, $|c| =1$, there exists non-orthogonal unit vectors $|\varphi_1\rangle,|\varphi_2 \rangle \in \hi_S$ such that $\langle  \varphi_1 | U \varphi_2 \rangle =0$. Make a counter-assumption that $ \langle  \varphi_1 | U \varphi_2 \rangle \neq 0$ whenever $\langle \varphi_1 | \varphi_2 \rangle \neq 0$. One can express $|\varphi_2\rangle = \alpha |\varphi_1\rangle + \beta |\eta\rangle$, where $\langle \varphi_1 | \eta\rangle =0$ and $|\alpha|^2 + |\beta|^2 =1$. Fix $|\varphi_1\rangle$ such that $|\langle \varphi_1| U \varphi_1\rangle| \neq 1$. We can then choose a unit vector $|\eta\rangle = \big(U^{\dagger} \varphi_1 - \langle \varphi_1| U^{\dagger} \varphi_1 \rangle \varphi_1 \big)/\sqrt{1-|\langle \varphi_1| U \varphi_1 \rangle |^2}$ orthogonal to $\varphi_1$. With the above choices we have $0\neq \langle  \varphi_1 | U \varphi_2 \rangle = \alpha \langle \varphi_1 | U \varphi_1 \rangle + \beta \sqrt{1-|\langle \varphi_1| U \varphi_1 \rangle|^2}$ for all $\alpha$ and $\beta$. Since $\alpha = \sqrt{1-|\langle \varphi_1| U \varphi_1 \rangle|^2}$ and $\beta = - \langle \varphi_1 | U \varphi_1\rangle$ contradict this while satisfying $|\alpha|^2 + |\beta|^2 =1$, the counter-assumption is falsified. Going the other way around, for any unit vector $|\varphi_1\rangle$ such that $|\langle \varphi_1| U \varphi_1\rangle| \neq 1$ we find that the unit vector 
\begin{align}\label{eq:varphi2}\nonumber
|\varphi_2\rangle &= \sqrt{1-|\langle \varphi_1 | U \varphi_1 \rangle |^2} \, |\varphi_1\rangle \\
&-\frac{\langle \varphi_1 | U \varphi_1 \rangle}{\sqrt{1-|\langle \varphi_1 | U \varphi_1 \rangle |^2}} \big( U^{\dagger} |\varphi_1\rangle - \langle \varphi_1 | U^{\dagger} \varphi_1 \rangle |\varphi_1\rangle\big)
\end{align} which satisfies $\langle \varphi_1 | U \varphi_2 \rangle =0$ and $\langle \varphi_1 | \varphi_2 \rangle \neq 0$. Fixing $U = U_1^{\dagger} U_2$ proves that $\mc F_{1/2} \big( \mc U_1, \mc U_2 \big) = 0$. For general $\alpha\in(0,1)$ the claim follows from the above considerations when noticing that $F_{\alpha}\big(\varrho_1,\varrho_2)=0$ if and only if $\varrho_1 \perp \varrho_2$; see Eq.\eqref{eq:defrenyi}.

\end{proof}

It should be stressed that for any pair of quantum channels $\mc E_i$, $i=1,2$ acting on a same Hilbert space one can find a processor $\langle \hi_E, \mc U\rangle$ from which the two channels can be induced. To see this, let us recall that any channel $\mc E:\trh \rightarrow \trh$ admits a Stinespring dilation of the form $\mc E(\varrho) = \text{tr}_{\hi_S} \left[ G \, \varrho \otimes |\eta \rangle \langle \eta | \, G^{\dagger}\right]$, where $G$ is a unitary on $\hi_S\otimes \hi_S$ and $|\eta\rangle\in\hi_S$ is some fixed unit vector. We can therefore assume without loss of generality that $\mc E_i(\varrho) = \text{tr}_{\hi_S}\left[ G_i \, \varrho \otimes |\eta \rangle \langle \eta | \, G_i^{\dagger}\right]$ for unitaries $G_i$, $i=1,2$. Consider any pair of orthogonal unit vectors $|\phi_1\rangle$ and $|\phi_2\rangle$, $\langle\phi_1| \phi_2\rangle =0$ in a two dimensional Hilbert space $\C^2$. Then it is easily verified that $U = G_1 \otimes |\eta\rangle\langle \eta| \otimes |\phi_1\rangle\langle\phi_1|+G_2 \otimes |\eta\rangle\langle \eta| \otimes |\phi_2\rangle\langle\phi_2|+\id_{\hi_S} \otimes \big( \id_{\hi_S}- |\eta\rangle\langle \eta|\big) \otimes \id_{\C^2}$ defines a unitary operator on $\hi_S \otimes \big( \hi_S \otimes \C^2\big)$. It follows that, by defining $\hi_E = \hi_S\otimes \C^2$ and $\mc U(\cdot) = U \cdot U^{\dagger}$, the vectors $|\eta \otimes \phi_i\rangle$ program the two channels $\mc E_i$, $i=1,2$, via processor $\langle \hi_E, \mc U\rangle$. As a result, any pair of quantum channels, even identical ones, can be programmed with orthogonal programming states. On the other hand, we showed above that for identical quantum channels the channel $\alpha$-fidelity reaches its maximum value 1. Importantly, this example alludes that is impossible to generally talk about the tightness of Ineq.\,\eqref{eq:progineq}.
% and that it quantifies the least wasteful way (in sense of orthogonality) in which one can program the corresponding channels.

In connection to the above, there are two complementary questions related to the tightness of Ineq.\,\eqref{eq:progineq} that we would like to address. Firstly, one may wonder if for any pair of programming states $\xi_i\in \mc S(\hi_E)$ there exists a processor $\langle \hi_E, \mc U\rangle$ realizing channels $\mc E_i$, $i=1,2$, respectively, that would satisfy $F_{\alpha} (\xi_1, \xi_2) = \mc F_{\alpha}\big(\mc E_1, \mc E_2 \big)$. The answer is affirmative, namely any quantum state $\xi\in \mc S(\hi_E)$ can be considered as the preparation channel $\mc E_\xi: \mc T(\hi_E) \to \C\{\xi\}$, $\mc E_\xi (\varrho) = \xi,$ for all $\varrho\in \mc S(\hi_E)$. Moreover, any such channel may be programmed with state $\xi\in \mc S(\hi_E)$ by using the processor $\langle \hi_E, \mc U_{\text{SWAP}} \rangle$, where $\mc U_{\text{SWAP}}(T_1 \otimes T_2) = T_2 \otimes T_1$, for all $T_i \in \mc T(\hi_E)$, that is $\mc E_\xi(\varrho) = \text{tr}_{\hi_E}\left[ \mc U_{\text{SWAP}}( \varrho \otimes \xi)\right]$, for all $\varrho\in \mc S(\hi_E)$. It is evident that in this case $\mc F_{\alpha} \big( \mc E_{\xi_1}, \mc E_{\xi_2} \big) = F_{\alpha} \big(\xi_1,\xi_2 \big)$, for all $\alpha\in(0,1)$. It is natural to ask if the contrary also holds, that is, if for any channels $\mc E_i$ one can find some processor and programming states $\xi_i$, $i=1,2$, respectively, such that $F_{\alpha} (\xi_1, \xi_2) = \mc F_{\alpha}\big(\mc E_1, \mc E_2 \big)$? It turns out that the answer in this case is negative. To confirm this, we recall that the programming states $\xi_i\in \mc S(\hi_E)$, $i=1,2,$ resulting into a unitary and any other different extremal channel are necessarily orthogonal to each other, that is, in particular, $F_{1/2} (\xi_1, \xi_2) =0$; see Ref.\,\cite[Prop.\,8]{HeiTuk2015}. It is, however, possible find a unitary channel $\mc U$ and an extremal channel $\mc E\neq \mc U$ with $\mc F_{1/2}\big(\mc U, \mc E) \neq 0$. Indeed, let us choose $\mc U$ as the one-qubit identity channel and for any $\varrho_{\vec{m}} \in \mc S(\C^2)$ define $\mc E(\varrho_{\vec{m}}) = \frac{1}{2}(\id + m_2 \, \sigma_2)$. The extremality of $\mc E$ can be verified by using results of Ref.\,\cite{Ruskai2002}, and furthermore using inequalities found in Ref.\,\cite{Miszczak09} one can shown that
\begin{eqnarray}
& &\mc F_{1/2}\big(\mc U, \mc E) \geq \nonumber \\
& &  \inf_{\substack{ ||\vec n||\leq 1, \\ ||\vec m||\leq 1\phantom{,}}} \frac{\tr{\varrho_{\vec n} \, \mc E(\varrho_{\vec{m}} )}}{\tr{\varrho_{\vec n} \, \varrho_{\vec m} } + \sqrt{(1- \tr{\varrho_{\vec n}^2}) (1- \tr{\varrho_{\vec m}^2})}} = 1/2\,. \nonumber \\
\end{eqnarray} We conclude that Ineq.\,\eqref{eq:progineq} fails to be tight in general.

\begin{figure}[t!]
\hspace{-1cm}\includegraphics[width=0.45\textwidth]{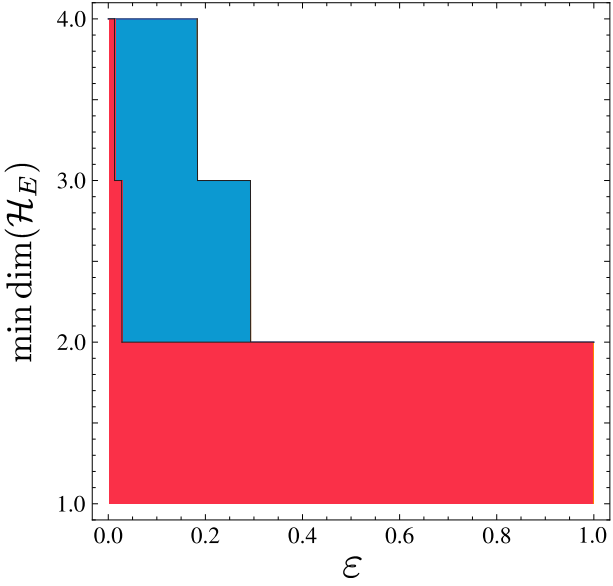}
\caption{{\bf Dimension of programmable processor in approximate programming.} We consider the dimension of a processor capable of implementing the four noisy unitary channels $\mc E_i: \varrho \mapsto (1-\varepsilon) \sigma_i \varrho \sigma_i^\dagger + \varepsilon \frac 12 \id,$ $i=0,\ldots, 3$, in terms of the noise parameter $\varepsilon\in[0,1]$. It may be confirmed that the process fidelity satisfies $F_{\text{proc}}\big(\sigma_i \cdot \sigma_i^\dagger, \mc E_i \big) \geq 1-\varepsilon$ for all $i= 0, \ldots, 3$. Accordingly, an analysis done in Ref.\,\cite{Hillery06} implies that any processor implementing $\mc E_i$, $i=0, \ldots, 3$, is at least four dimensional for $0 \leq \varepsilon < \left[3 (13 + 2 \sqrt{42}) \right]^{-1}$, at least three dimensional for $\left[3 (13 + 2 \sqrt{42}) \right]^{-1} \leq \varepsilon < \left[2 (9 + 4 \sqrt{5} )\right]^{-1}$ and at least two dimensional for $\left[2 (9 + 4 \sqrt{5}) \right]^{-1} \leq \varepsilon< 1$ excluding the dimensions within the red area. However, our approach results to tighter limits which we have denoted in blue: the processor is necessarily at least four dimensional for $0\leq \varepsilon < \frac13 (3 - \sqrt 6)$, at least three dimensional for $\frac13 (3 - \sqrt 6) \leq \varepsilon < \frac12 (2 - \sqrt 2)$ and at least two dimensional for $\frac12 (2 - \sqrt 2) \leq \varepsilon< 1$. } 
\label{fig:apprprog}
\end{figure}
Nevertheless, the channel $\alpha$-fidelity $\mc F_\alpha$ is a genuinely important figure of merit in quantum programming, as we will see in the following. According to the previous proposition, programming $N$ different unitary channels requires $\dim(\hi_E) \geq N$. The processor's dimension may, however, be lowered if some error is accepted. Let us denote $d=\dim(\hi_S)$. Suppose we wish to program the approximate unitary channels $\mc E_{U_i}(\varrho) = (1-\varepsilon) \mc U_j(\varrho) + \varepsilon \frac{1}{d} \id_{\hi_S}$ with programming vectors $|\phi_i\rangle$, respectively, where $U_i$, $i=1,\ldots, N$ are unequal unitaries and $\varepsilon \in [0,1]$ describes a fixed error rate. From Ref.\,\cite{Miszczak09} we find that $F_{1/2}\big(\mc E_{U_j}(\varrho_1),\mc E_{U_k}(\varrho_2)\big) \leq \tr{\mc E_{U_j}(\varrho_1)\,\mc E_{U_k}(\varrho_2)} + \sqrt{\big(1-\tr{\mc E_{U_j}(\varrho_1)^2}\big)\big(1-\tr{\mc E_{U_k}(\varrho_2)^2}\big)}.$ Choosing $\varrho_i = |\varphi_i\rangle \langle \varphi_i|$, $i=1,2$ simplifies this inequality yielding $\frac{F_{1/2}\big(\mc E_{U_j}(\varrho_1),\mc E_{U_k}(\varrho_2)\big)}{F_{1/2}(\varrho_1, \varrho_2)} \leq \frac{(1-\varepsilon)^2 |\langle U_j \varphi_1| U_k \varphi_2 \rangle|^2}{|\langle \varphi_1| \varphi_2 \rangle|}+\frac{(2-\varepsilon)\varepsilon}{|\langle \varphi_1| \varphi_2 \rangle|}.$
In particular, choosing $\varphi_2$ as in Eq.\,\eqref{eq:varphi2} for $U=U_j^{\dagger}U_k$ implies 
\begin{eqnarray}\label{eq:approbound}
F_{1/2} (\phi_j,\phi_k) &\leq & \inf_{|\varphi_1\rangle} \frac{(2-\varepsilon)\varepsilon}{\sqrt{1-|\langle U_j \varphi_1| U_k \varphi_1 \rangle|^2}} \nonumber \\
& =& \frac{(2-\varepsilon)\varepsilon}{\sqrt{1-\inf_\varrho F_{1/2}\big(\mc U_j(\varrho),\mc U_k(\varrho)\big) }} \doteq g_{jk}(\varepsilon), \nonumber \\
\end{eqnarray} 
when $j\neq k$. We may use this bound to solve for the largest set of linearly independent programming vectors using the following result: if $|\phi_j\rangle,\, j=1,\ldots,K$, are unit vectors and $F_{1/2}(\phi_j, \phi_k) \leq \frac{1}{K-1}$ whenever $j\neq k$, then the vectors $|\phi_j\rangle, \, j=1,\ldots,K$, are linearly independent \cite{Hillery02}. Let us denote with $K_\varepsilon$ the largest integer such that $K_\varepsilon < \max_{j\neq k \in\{1,\ldots N\}}\big(1/g_{jk}(\varepsilon)+1\big)$. The previous result then implies that any set of vectors whose size is less than or equal to $K_\varepsilon$ is linearly independent. Therefore, if $N \leq K_\varepsilon$, then all of the programming vectors are linearly independent and $\dim(\hi_E) \geq N$. On the other hand if $N> K_\varepsilon$, then $\dim(\hi_E) \geq K_\varepsilon$ \cite{Hillery06}.

As an example, let us consider approximate programming of qubit unitaries $\sigma_i$, $i=0,\ldots,3$, where $\sigma_0 = \id_{\C^2}$. For all pairs the quantity $\inf_\varrho F_{1/2}\big(\sigma_j \varrho \sigma_j,\sigma_k \varrho \sigma_k\big)$ vanishes. Therefore, according to Eq.\,\eqref{eq:approbound} the processor in approximate programming of the above unitaries is at least four dimensional for $0\leq \varepsilon < \frac13 (3 - \sqrt 6)$, at least three dimensional for $\frac13 (3 - \sqrt 6) \leq \varepsilon < \frac12 (2 - \sqrt 2)$ and at least two dimensional for $\frac12 (2 - \sqrt 2) \leq \varepsilon< 1$. In Fig.\,\ref{fig:apprprog} these bounds have been compared to those previously solved in Ref.\,\cite{Hillery06}. It is noteworthy that our approach gives tighter limits for the amount of noise necessary to be present in order to approximately implement the four unitary channels $\sigma_i$, $i=0,\ldots,3$.

%%%%%%%%%%%%%%%%%%%%%%%%%%%%%%%
\subsection{Ruling out incompatible environmental properties} Let us consider a general thermal environment $\hi_E$ and two initial states of the environment in different temperatures $T_1$ and $T_2$. The environmental states therefore attain the Gibbs form $\xi(\beta_i) \doteq \expo{- \beta_i \Ho_{E}} / Z(\beta_i, \Ho_E)$, $i=1,2$, where $\beta_i = (k_B T_i)^{-1}$, where $\Ho_{E}$ is a Hamiltonian of the environment and $Z(\beta_i,\Ho_{E}) \doteq \tr{ \expo{ -\beta_i \Ho_{E}}}$ are the corresponding partition functions for $i=1,2$. 
% Here $\Ho_{E}$ is the Hamiltonian of the environment with a spectral decomposition $\Ho_{ E} = \sum_j \varepsilon_j |j\rangle \langle j|$ and $Z(\beta_i,\Ho_{E}) \doteq \tr{ \expo{ -\beta_i \Ho_{E}}} = \sum_j \expo{-\beta_i \varepsilon_j}$ are the corresponding partition functions for $i=1,2$.

	Our goal is to exclude some properties of environment inscribed in the spectrum of the Hamiltonian $\Ho_E$ that are incompatible with Ineq.\,\eqref{eq:mainineq}. Since for thermal states $\left[\xi(\beta_1),\xi(\beta_2)\right]=0$, we easily solve
\begin{widetext}
\begin{equation}
S_\alpha\big(\xi(\beta_1)||\xi(\beta_2)\big) = \frac{1}{\alpha -1} \big( \ln\left[ Z(\alpha \beta_1 + (1-\alpha) \beta_2,\Ho_{E}) \right]
- \alpha \ln\left[ Z( \beta_1,\Ho_{E}) \right] -  (1-\alpha) \ln\left[ Z( \beta_2,\Ho_{E} ) \right] \big). 
\end{equation}
In particular, for all $\alpha \in [\frac{1}{2}, 1)$ Ineq.\,(\ref{eq:mainineq}) implies the following inequality 
\begin{equation}
\ln\left[ Z(\alpha \beta_1 + (1-\alpha) \beta_2,\Ho_{E}) \right]- \alpha \ln\left[ Z( \beta_1,\Ho_{E}) \right]
-  (1-\alpha) \ln\left[ Z( \beta_2,\Ho_{E} ) \right]%\nonumber \\
\leq  \inf_{t\geq 0} \ln\left[\mc F_{\alpha} \big(\mc E_1^{(t)},\mc E_2^{(t)} \big)\right].
\end{equation}
In fact, due to the commutativity of the thermal states, we can expand this for $\alpha \in (0, 1)$ such that 
\begin{equation}\label{eq:mainineq3}
\ln\left[ Z(\alpha \beta_1 + (1-\alpha) \beta_2,\Ho_{E}) \right]- \alpha \ln\left[ Z( \beta_1,\Ho_{E}) \right]  
-(1-\alpha) \ln\left[ Z( \beta_2,\Ho_{E} ) \right]  \leq  \inf_{t\geq 0} \left\{ \begin{array}{cc}
\ln[\mc F_{\alpha} \big(\mc E_2^{(t)},\mc E_1^{(t)}\big) ],  \text{for} \  \alpha \in (0, \frac{1}{2}) \\
\ln[\mc F_{\alpha} \big(\mc E_1^{(t)},\mc E_2^{(t)}\big)],  \text{for} \ \alpha \in [\frac{1}{2}, 1)
\end{array} \right. .
\end{equation}
%Because of strict concavity of the logarithm function, $\ln [ t x + (1-t) y ] > t \ln[ x ]+ (1-t) \ln [ y ]$ for all $t \in [0,1)$, $0< x,y \in \R$, $x\neq y$, we conclude that the left-hand-side (and consequently also the right-hand-side) of the above inequality is strictly positive for all $\alpha \in [0,1)$ and $\beta_1\neq \beta_2$.
\end{widetext}

Let us consider the implications of the above inequalities. Assume that it is possible to prepare the environment in two different known temperatures and
perform the full-process tomography of the two induced dynamics.
%measure the induced dynamics.
This enables us to determine the values of $\mc F_{\alpha}$ in the right-hand side of Ineq.\,(\ref{eq:mainineq3}). On the other hand, any hypothesis about the Hamiltonian of the environment determines the left-hand side of the inequality. If the inequality is not satisfied, the hypothesized Hamiltonian of the system can be ruled out as incompatible with the induced pair of dynamics. We highlight that in order to do so neither the system-environment coupling nor the Hamiltonian of the system needs to be specified.

As a demonstration of the power of this method, let us consider a specific example of a qubit system coupled to an environment consisting of a single harmonic oscillator: such a situation could occur, for instance, when a two-level atom is passing through an optical cavity including only a single quantized mode. Assume that the mechanism of interaction is unknown as well as the oscillator frequency which we nevertheless wish to determine. Mathematically, we only know that the Hamiltonian of the environment is of the form $\Ho_E=\omega \big(b^{\dagger} b + \frac{1}{2}\id\big)$. Choosing different values for $\omega$ one plots the dashed coloured curves in Fig.\,\ref{fig:hamiltonians} given by the left-hand side of Ineq.\,\eqref{eq:mainineq3}. For a concrete example of the (exactly solvable) reduced dynamics giving the right-hand-side of Ineq.\,(\ref{eq:mainineq3}) and the solid black curve in Fig.\,\ref{fig:hamiltonians}, we consider the interaction model of the form
\begin{eqnarray}
 \Ho	&=& \Ho_S + \Ho_E + \Ho_I\nonumber \\
	&=& \frac{\omega_0}{2}\sigma_3 + \omega \big(b^{\dagger} b + \frac{1}{2}\id\big) + \sigma_3\otimes \big(g b^\dagger + g^* b\big)\,
\end{eqnarray}
analysed in \cite{Breuer}. One can see, that the inequality forbids frequency values greater than 3.1 times the real frequency $\omega$.
%The energy-eigenvalues of the Hamiltonian $\Ho_{E}=\sum_j \varepsilon_j |j\rangle\langle j |$ then read $\varepsilon_j = \hbar \omega (j +1/2)$, $j=0,1,2,\dots$, where $\omega$ is the frequency of the oscillator. The above developed method enables us to exclude some values of $\omega$ incompatible with Ineq.\,(\ref{eq:mainineq3}); see Fig.\,\ref{fig:hamiltonians}.
For details we refer the reader to Appendix \ref{app:b}. 
\begin{figure}[t!]
\includegraphics[width=0.45\textwidth]{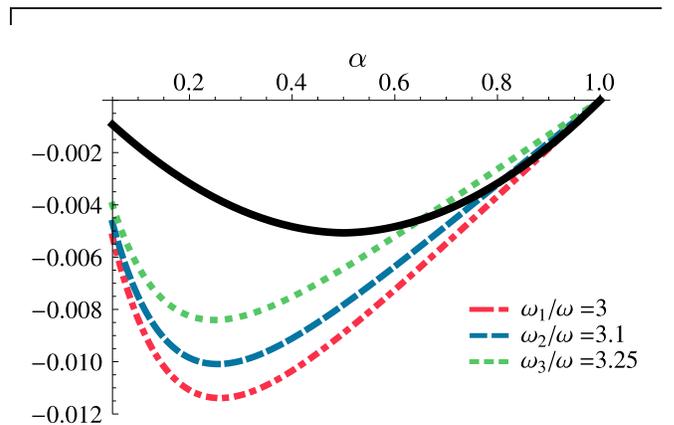}
\caption{ {\bf Excluding the frequency values of a harmonic oscillator.} (In this example we have fixed $\hbar=1=k_B$.) We prepare the environment in two temperatures and measure the induced dynamics, which gives us right-hand-side of Ineq.\,(\ref{eq:mainineq3}) (black line). The three dashed coloured lines correspond to the left-hand-side of Ineq.\,\eqref{eq:mainineq3} with different choices for $\omega_k$, $k=1,2,3$, in units of the actual frequency $\omega$. Knowing the two temperatures of the environment ($T_1/\omega=0.25$) and ($T_2/\omega =0.75$) we notice that some frequencies ($\omega_1/\omega=3$) are compatible with Ineq.\,(\ref{eq:mainineq3}) while others ($\omega_2/\omega=3.1$ and $\omega_3/\omega=3.25$) violate it. 
In fact, we see that Ineq.\,(\ref{eq:mainineq3}) is violated for all frequencies larger than $\omega_2/\omega=3.1$, and thus $\omega_2/\omega=3.1$ is the crossover between the compatible and incompatible frequencies.
As the probe we use a qubit and the induced dynamics result from an exactly solvable model described in \cite{Breuer}. The optimal initial state(s) of the qubit are $\varrho = |+\rangle \langle +|$ with $|+\rangle = \frac{1}{\sqrt{2}}(|0\rangle + |1\rangle)$ being the eigenstate of Pauli operator $\sigma_1$; for details see Appendix B.}\label{fig:hamiltonians}
\end{figure}

%%%%%%%%%%%%%%%%%%%%%%%%%%%%%%%
\subsection{Quantum thermometry} Let us consider the situation of the previous example from a different perspective. Suppose now that a thermal environment described by a single harmonic oscillator with Hamiltonian $\Ho_E=\omega \big(b^{\dagger} b + \frac{1}{2}\id\big)$ is initially prepared at calibration temperature $T_0 = 0$, that is $\xi(\beta_0=\infty) = |0 \rangle \langle 0|$. Let us assume we are tasked with probing another temperature $T$ of the environment. The $\alpha$-R\'{e}nyi divergence between the states $\xi(\beta_0)$ and $\xi(\beta=1/(k_B T) )$ now reads $S_\alpha \big(\xi(\beta_0)|| \xi(\beta)\big) = \frac{1-\alpha}{\alpha - 1} \big(-\frac{\beta \hbar \omega}{2} - \ln[ Z(\beta, \Ho_{E}) ] \big) $ and therefore
%\begin{eqnarray}
%\ln\big[1-\expo{ -\hbar \omega \beta }\big] \leq 
%\left\{ \begin{array}{ll}
%\ln\left[\mc F_{\alpha} \big(\mc E_T ,\mc E_0 \big)\right], & \text{for} \  \alpha \in %(0, \frac{1}{2})\\
%\frac{1}{1-\alpha}\ln\left[\mc F_{\alpha} \big(\mc E_0, \mc E_T \big)\right], & %\text{for} \ \alpha \in [\frac{1}{2}, 1) 
%\end{array} \right. ,
%\end{eqnarray}
%for all $\alpha \in (0,1)$.  
after some simple algebra we solve from Ineq.\eqref{eq:mainineq1} a lower bound for the temperature
\begin{align}\nonumber
&\frac{k_B T}{\hbar \omega}  \geq \\ \label{eq:thermoineq1}
&-1 \times  
\left\{ \begin{array}{ll}
\left[\ln\left[1-\inf_{t\geq 0} \mc F_{\alpha} \big(\mc E_T^{(t)}, \mc E_0^{(t)} \big)^{\frac{1}{\alpha}}\right]\right]^{-1}, &  \alpha \in (0, \frac{1}{2}) \\
\left[\ln\left[1- \inf_{t\geq 0} \mc F_{\alpha} \big(\mc E_0^{(t)} ,\mc E_T^{(t)} \big)^{\frac{1}{1-\alpha}}\right]\right]^{-1}, &   \alpha \in [\frac{1}{2}, 1) 
\end{array} \right. % \qquad 
%\frac{k T}{\hbar \omega}  \geq  \inf_{\alpha\in[0,1)} \left[\ln\big( \expo{\frac{1}{1-\alpha} G(\alpha, \infty, \beta)} +1 \big)\right]^{-1}.
\end{align}
for all $T\geq 0$. It can be confirmed, that with temperatures lower than a limiting temperature $(k_B T)/ (\hbar \omega) \lesssim 1.03$ the term $-\ln[ Z(\beta, \Ho_{E}) ] = \frac{\beta \hbar \omega}{2} + \ln[1-e^{-\beta \hbar \omega}]$ is non-negative, from which using Ineq.\,\eqref{eq:mainineq1} we may also obtain an upper bound for the temperature
\begin{align} \nonumber
&\frac{k_B T}{\hbar \omega} \leq   \\ \label{eq:thermoineq2}
& -1/2 \times 
\left\{ \begin{array}{ll}
\left[\ln\left[\inf_{t\geq 0} \mc F_{\alpha} \big(\mc E_T^{(t)}, \mc E_0^{(t)} \big)^{\frac{1}{\alpha}}\right]\right]^{-1}, & \alpha \in (0, \frac{1}{2}) \\
\left[\ln\left[\inf_{t\geq 0} \mc F_{\alpha} \big(\mc E_0^{(t)} ,\mc E_T^{(t)} \big)^{\frac{1}{1-\alpha}}\right]\right]^{-1}, &  \alpha \in [\frac{1}{2}, 1)
\end{array} \right. \qquad
%\frac{k T}{\hbar \omega}  \geq  \inf_{\alpha\in[0,1)} \left[\ln\big( \expo{\frac{1}{1-\alpha} G(\alpha, \infty, \beta)} +1 \big)\right]^{-1} \,.
\end{align}

Since there is no dependency on the parameter $\alpha$ on the left-hand-side of the above inequalities, we optimize over this parameter. The quantities on the right-hand-side of the inequalities above can be solved (numerically) after measuring the dynamics; these observations allow us to estimate the unknown temperature $T$.

The inequalities (\ref{eq:thermoineq1}) and (\ref{eq:thermoineq2}) may, unfortunately, be difficult to evaluate in general due to double optimization with respect to states and therefore a more accessible form is desirable. To achieve this, we use the fact that $\mc F_{\alpha} \big(\mc E_1, \mc E_2\big) \leq \inf_\varrho F_{\alpha}\big(\mc E_1(\varrho), \mc E_2(\varrho)\big)$, for all $\alpha\in(0,1)$. Since $S_{\alpha}(\cdot || \cdot)$ is monotonically increasing in $\alpha$ \cite{Muller2013}, that is, $S_{\alpha}\big(\varrho_1 || \varrho_2\big)\leq S_{\alpha'}\big(\varrho_1 || \varrho_2\big)$ whenever $\alpha \leq \alpha'$, and since the logarithm function is monotonically increasing, we conclude that $F_{\alpha} ( \varrho_1, \varrho_2)^{\frac{1}{1-\alpha}}$ monotonically decreases in $\alpha$. On the other hand, since for all positive semidefinite linear operators $A,B$ and numbers $q\geq 0, r \geq 1$ the Araki-Lieb-Thirring inequality \cite{Lieb1976, Araki1990} states that $\tr{\big(B^{1/2} A B^{1/2}\big)^{rq}} \leq \tr{\big(B^{r/2} A^r B^{r/2}\big)^q}$ holds and since $F_{\alpha}(\varrho_1,\varrho_2)$ may be recast in the form \footnote{This follows from the fact that the operators $X X^\dagger$ and $X^\dagger X$ have the same non-zero eigenvalues, and therefore $\tr{(X X^\dagger)^\alpha} =\tr{(X^\dagger X)^\alpha}$.} $F_{\alpha}(\varrho_1,\varrho_2) = \tr{\big( \varrho_1^{1/2} \varrho_2^{(1-\alpha)/\alpha} \varrho_1^{1/2} \big)^{\alpha}}$, we have 
\begin{eqnarray}\nonumber
F_{\alpha}(\varrho_1,\varrho_2)&=&\tr{\big( \varrho_1^{1/2} \varrho_2^{(1-\alpha)/\alpha} \varrho_1^{1/2} \big)^{\frac{\alpha}{1-\alpha}(1-\alpha)}} \\ \nonumber
&\leq& \tr{\big( \varrho_1^{\alpha/2(1-\alpha)} \varrho_2 \varrho_1^{\alpha/2(1-\alpha)} \big)^{1-\alpha}}  \\
&=& F_{1-\alpha}(\varrho_2,\varrho_1)
\end{eqnarray} and therefore $F_{\alpha}(\varrho_1,\varrho_2)^{1/(1-\alpha)} \leq F_{1-\alpha}(\varrho_2,\varrho_1)^{1/(1-\alpha)}$ whenever $\alpha\geq 1/2$ ($\Leftrightarrow r=\frac{\alpha}{1-\alpha} \geq 1$). To summarize
 \begin{eqnarray} \label{eq:thermoineq3}
\frac{k_B T}{\hbar \omega}  &\geq & - \ln\left[1-\inf_{\varrho,t\geq 0}\lim_{\alpha \nearrow 1}   F_{\alpha} \big(\mc E_0^{(t)}(\varrho), \mc E_T^{(t)}(\varrho) \big)^{\frac{1}{1-\alpha}}\right]^{-1} \nonumber \\
&=& - \ln\left[1- e^{-\sup_{\varrho,t\geq 0} S_1\big(\mc E_0^{(t)}(\varrho)||\mc E_T^{(t)}(\varrho)\big)}\right]^{-1}\,,% \qquad T\geq 0,
 \end{eqnarray} 
for $T\geq 0$, where $\lim_{\alpha \nearrow 1} S_\alpha (\varrho_1||\varrho_2 \big) = S_1\big(\varrho_1||\varrho_2 \big) \doteq \tr{ \varrho_1\big(\ln\varrho_1 - \ln\varrho_2 \big)}$ is the Kullback-Leibler divergence \cite{Muller2013}. Similarly, for the upper bound of temperature one derives
\begin{eqnarray} \label{eq:thermoineq4}
\frac{k_B T}{\hbar \omega}  &\leq & -1/2  \ln\left[\inf_{\varrho,t\geq 0} \lim_{\alpha \nearrow 1}  F_{\alpha} \big(\mc E_0^{(t)}(\varrho), \mc E_T^{(t)}(\varrho) \big)^{\frac{1}{1-\alpha}}\right]^{-1} \nonumber \\
&=&\left( 2 \sup_{\varrho,t\geq 0} S_1\big(\mc E_0^{(t)}(\varrho)||\mc E_T^{(t)}(\varrho)\big) \right)^{-1}\,,
\end{eqnarray}
when $0 \leq (k_B T)/ (\hbar \omega) \lesssim 1.03.$
\begin{figure}[t!]
\includegraphics[width=0.45\textwidth]{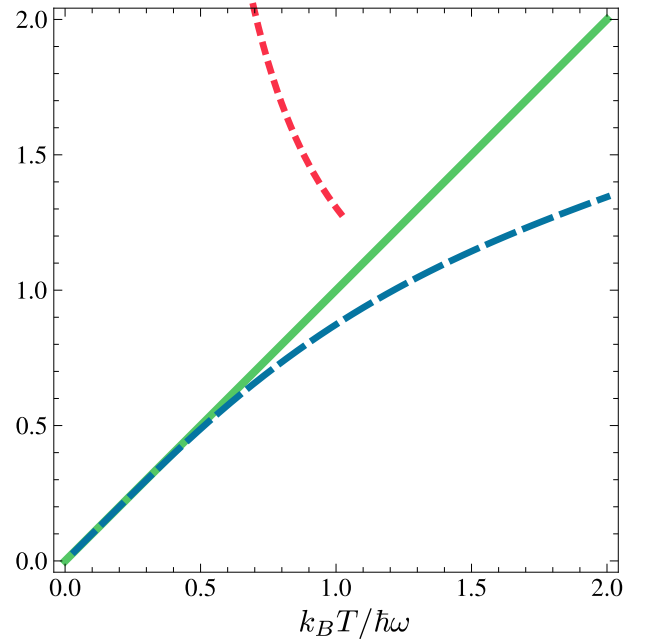}
\caption{ {\bf Estimation of temperature.} We assume that the environment is prepared in two different temperatures ($(k_B T_0)/(\hbar \omega) = 0$ and $(k_B T)/(\hbar \omega)$) and measure the induced dynamics. The green line represents the actual temperature $(k_B T)/(\hbar \omega)$. From Ineq.\,\eqref{eq:thermoineq3} we solve a lower bound of $(k_B T)/(\hbar \omega)$, plotted above as the blue dashed curve. The red dotted curve corresponds to the upper bound of the temperature valid whenever $(k_B T)/(\hbar \omega) \lesssim 1.03$ given by Ineq.\,\eqref{eq:thermoineq4}. The system that we use for probing the temperature is a qubit and the dynamics are due to Jaynes-Cummings model in the resonant case (see Ref.\,\cite{Smirne2010}). The initial state of the qubit that we have used for in the plot is $\varrho=|+\rangle \langle + |$, since numerics suggest that it leads to the optimal bounds. } \label{fig:thermo}
\end{figure}
%The latter bound corresponds to the worst case $\alpha=1/2$. Interestingly, $\inf_\varrho F_{1/2}(\mc E_0^{(t)}(\varrho), \mc E_T^{(t)}(\varrho) )$ appearing in this bound is a well-known quantity in quantum information science called {\it(minimum) gate fidelity} \cite{NielsenInfo}.

In Fig.\,\ref{fig:thermo} we have plotted $(k_B T)/(\hbar \omega)$ and the optimal bounds given by the right-hand-sides of Ineqs.\,(\ref{eq:thermoineq3}) and (\ref{eq:thermoineq4}) for the resonant Jaynes-Cummings model, which describes a two-level atom interacting with a single quantized mode of radiation in an optical cavity \cite{Breuer, Smirne2010}; for details we refer the reader to Appendix C. In the figure, we have chosen the initial state of the qubit probe as $\varrho=|+\rangle \langle + |$, since numerics suggest that it leads to the optimal bounds. 
%In fact, whenever a coupling leads to pure dephasing dynamics of a qubit this choice is the one minimizing the $\alpha$-fidelity of channels, as has been alluded in the Appendix A. It is an interesting consequence, that for dephasing dynamics the $\alpha$-fidelity of channels becomes continuous: the closer the dynamics are to each other, the larger their $\alpha$-fidelity and vice versa. Accordingly, for dephasing channels, the lower bound given by our protocol will become better as the temperature gets smaller: this behaviour is also clearly seen in Fig.\,\ref{fig:thermo}.

A standard procedure in (quantum) thermometry is to bring the thermometer into contact with a bath, let it thermalize and then read out the temperature. In the quantum scenario, we say that a qubit probe has {\it thermalized} with the bath if any initial state $\varrho$ of the probe evolves in the long time limit $t\to \infty$ into the equilibrium Gibbs state
$\varrho \to \varrho_{\text{eq}} \doteq  \big(1-p(T) \big) \, \ket 0 \bra 0 + p(T) \, \ket 1 \bra 1$, where $p(T) = \big[ 1 +  \text{exp}(\hbar \omega/k_B T ) \big]^{-1}$. %CITES
Quantum thermometry protocols  \cite{Brunelli2012,Higgins2013,Jevtic2015,Johnson2016} and optimality of initial probe states \cite{Correa2015} have been studied for both fully and partially thermalizing probes. To our best knowledge all pre-existing protocols are based on knowing the exact form of the coupling. Also the optimal initial probe states have been solved in the partially thermalized probes only for fixed system-probe couplings.

In the case of fully thermalized probes, calculating the $\alpha$-fidelity becomes particularly easy: $\lim_{t\to \infty} \mc F_{\alpha} \big(\mc E_0^{(t)}, \mc E_T^{(t)} \big) = \big(1-p(T) \big)^{1-\alpha}$. Accordingly, our protocol sets the limits 
\begin{eqnarray}
- \big[ \ln[ p(T) ] \big]^{-1} \leq \frac{k_B T}{\hbar \omega} \leq - 1/2 \big[ \ln[ 1- p(T) ] \big]^{-1},
\end{eqnarray}
with the upper bound valid whenever $0 \leq (k_B T)/ (\hbar \omega) \lesssim 1.03.$ It is noteworthy, that due to the special form of the state $\varrho_{\text{eq}}$ the temperature $T$ could be solved analytically from the state occupations, which shows that our protocol is the not the optimal one for thermometry under the assumption of thermalization. However, there is no {\it a priori} need to require thermalization of the probe with the bath in our protocol. This is an important advantage, since thermalization itself in a fully quantum scenario may not be a well-defined concept due to the ever present quantum fluctuations and the long time-scales needed compared to the survival time of the quantum properties of the bath.

We wish to emphasize that knowing the specific system environment coupling is not required in the above two parameter estimation protocols {\bf B} and {\bf C}: to the best of our knowledge, such probing approaches independent of the description of the coupling have never been proposed before. While the precisions of the frequency estimation and thermometer are clearly not optimal, these examples serve as proofs-of-principle demonstrations of the ability to extract information on certain physical quantities by means of our approach, making only minimal assumptions on the microscopic details of the system to be probed. 
From a broader point of view, the two protocols {\bf B} and {\bf C} could be applied to estimate, not only the frequency and temperature, but also other physical quantities of a system of interest. The task to probe some of such quantities could be unattainable ''conventionally'', since the description of the coupling model between the system and the probe is not either accurately known or is borderline totally missing. An ultimate example of such case is the quantum gravity, since no commonly agreed model of gravitational interaction in quantum regime exists.

\subsection{Bounds for the Loschmidt Echo}

In our last application we change the perspective altogether and consider a different implication of Ineq.\,\eqref{eq:mainineq2}. We show how from the $\alpha$-fidelity of the initial environmental states we can infer certain properties of the induced dynamics, such as information back-flow (non-Markovianity). Importantly, back-flow of information can be used to protect and restore vital quantum properties of the system subjected to detrimental noise \cite{Rivas2014, Breuer2016, Hinarejos2016}.

As an example we concentrate on one-qubit dephasing channels arising from system environment interaction. Although the method of this subsection holds generally for such channels, we demonstrate the power of the protocol by focusing on a specific situation with known solutions only for certain initial states of environment: a qubit transversely coupled to its environment, which is an Ising spin chain in a transverse field. 
The total evolution of such a composite system is governed by the Hamiltonian $\Ho(\lambda,\delta) = - J \,\Sigma_j\big( \sigma_3^{(j)} \otimes \sigma_3^{(j+1)} + \lambda \, \sigma_1^{(j)} + \delta \, |e\rangle\langle e| \otimes \sigma_1^{(j)} \big)$. Here the parameters $J$, $\lambda$ and $\delta$ characterize the strength of the nearest neighbour interactions in the Ising chain, its coupling to a transverse field and the coupling between system and environment, respectively \cite{Quan2006, Haikka2012}. It is known, that the initial environmental state $|\phi\rangle$ induces dynamics corresponding to a pure dephasing channel $
\mc E^{(t)}_\phi(\varrho) = p_\phi(t) \varrho + \big(1-p_\phi(t)\big) \sigma_3 \varrho \sigma_3$ with the probability $p_\phi(t)$ given by
\begin{eqnarray}
p_\phi(t)=\frac{1}{2}\big(1+\sqrt{L_\phi(\lambda,t)}\big),
\end{eqnarray}
where $L_\phi(\lambda,t) = |\langle \phi | e^{i \Ho(\lambda,0) t} e^{-i \Ho(\lambda+\delta,0) t} \phi \rangle|^2$ is the Loschmidt echo corresponding to the state $|\phi\rangle$. Despite the apparent simplicity, finding an analytical expression of $L_\phi(\lambda,t)$ is a difficult task for a general initial state $|\phi\rangle$. In fact to our best knowledge, the only known analytical solution exists for the ground state $|\phi_0\rangle$ of the Hamiltonian $\Ho(\lambda,0)$ for which one obtains
\begin{eqnarray}
L_{\phi_0}(\lambda,t)= \Pi_{k>0} \big(1-\sin^2(2 \alpha_k) \sin^2(\varepsilon_k t) \big),
\end{eqnarray}
where $\alpha_k$ are the Bogoliubov angles and $\varepsilon_k$ the single quasiparticle excitation energies of the system with the qubit in the excited state $|e\rangle$ \cite{Quan2006, Haikka2012}.
\begin{figure*}[t!]
\begin{minipage}[t]{0.9\textwidth}
\includegraphics[width=0.9\textwidth]{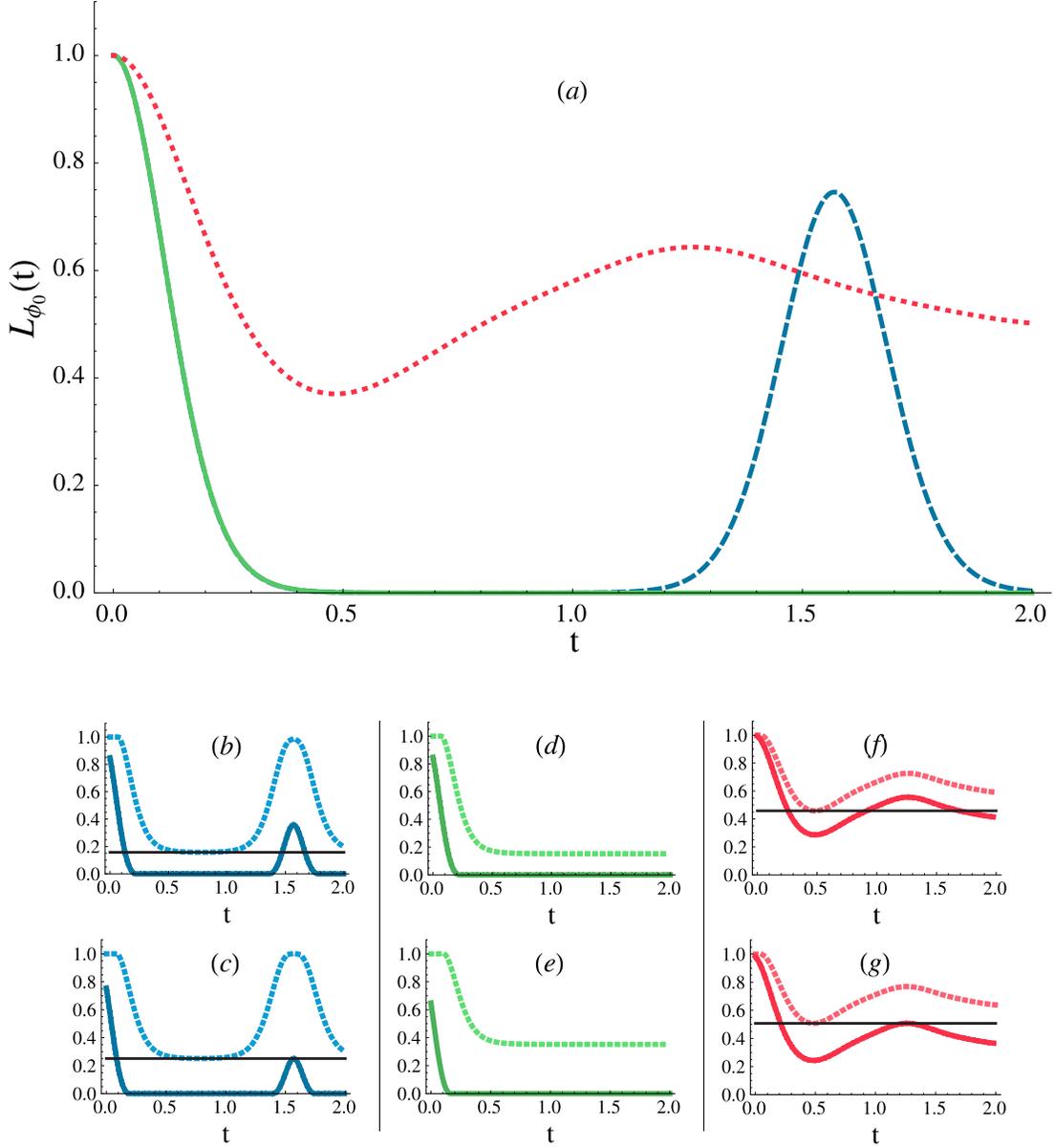}
\end{minipage}
\caption{{\bf Bounding the Loschmidt echo.} (a) The dynamics of the Loschmidt echo induced by the environmental ground state $\phi_0$. The parameter values of the Hamiltonian $\Ho(\lambda)$ are set as $J = 1,\, \delta = 0.1$ and $N = 4000$. The blue dashed line corresponds to the choice $\lambda = 0.01$, the green solid line corresponds to the critical value $\lambda = 0.9$ and the red dotted line corresponds to the value $\lambda = 1.8$ (b) - (g) The dynamics of the upper and lower bound of the unknown Loschmidt echo are represented by the curves in lighter and darker hue of the corresponding color choices of the parameter $\lambda$, respectively. The black solid line indicates the level of the first local minimum of the upper bound. In (b) and (f) the lower bounds cross the black line after the local minimum of the upper bound, which implies revivals in the Loschmidt echo and thus information back-flow.
(b) $\lambda = 0.01$, $F_{1/2}(\phi_0, \phi) = 0.98$.
(c) $\lambda = 0.01$, $F_{1/2}(\phi_0, \phi) = 0.966675$.
(d) $\lambda = 0.9$, $F_{1/2}(\phi_0, \phi) = 0.98$.
(e) $\lambda = 0.9$, $F_{1/2}(\phi_0, \phi) = 0.95$.
(f) $\lambda = 1.8$, $F_{1/2}(\phi_0, \phi) = 0.999$.
(g) $\lambda = 1.8$, $F_{1/2}(\phi_0, \phi) = 0.99761$.
} \label{fig:loschmidt}
\end{figure*}

Using Ineq.\,\eqref{eq:mainineq2} allows us to give an estimate for the Loschmidt echo $L_\phi(\lambda,t)$ that is valid for any state $|\phi\rangle$. Indeed, fixing an initial state $\varrho = |+\rangle\langle +|$ of the qubit, $|+\rangle = \frac{1}{\sqrt{2}}( |0\rangle + |1\rangle)$ being the eigenstate of $\sigma_1$, we can solve 
\begin{eqnarray}\label{eq:lossi}
\mc E^{(t)}_\phi(\varrho) &= &\frac{1}{2} \left( 
\begin{array}{cc} 1 & \sqrt{L_\phi(\lambda,t)} \\
\sqrt{L_\phi(\lambda,t)} & 1
\end{array}
\right).
\end{eqnarray}
It can be concluded that 
\begin{align} \nonumber
F_{1/2}\big( \phi_0, \phi\big)&\leq  F_{1/2}\big(\mc E^{(t)}_{\phi_0}(\varrho), \mc E^{(t)}_\phi(\varrho)\big)\nonumber\\ \nonumber
&= \frac{1}{2} \sqrt{\big(1-\sqrt{L_{\phi_0}(\lambda,t)}\big)\big(1-\sqrt{L_\phi(\lambda,t)}\big)}\nonumber\\ \label{eq:loschmidt}
&+\frac{1}{2} \sqrt{\big(1+\sqrt{L_{\phi_0}(\lambda,t)}\big)\big(1+\sqrt{L_\phi(\lambda,t)}\big)},
\end{align}
from which an estimate for $L_\phi(\lambda,t)$ in terms of $F_{1/2}\big( \phi_0, \phi\big)$ may be solved. In Fig.\,\ref{fig:loschmidt} we have presented the results for three different values of the strength of the transverse field $\lambda$ fixing the parameters $\delta = 0.1$, $J = 1$ and the number of spins $N = 4000$. In particular, we see that for those $\lambda$ resulting to 
%the non-Markovian behaviour 
revivals in Loschmidt echo
in the ground state $|\phi_0\rangle$, after a certain limiting value of $F_{\text{lim}}\leq F_{1/2}\big( \phi_0, \phi\big)$  such revivals are guaranteed %non-Markovianity
for the state $|\phi\rangle$ also. The revivals of Loschmidt echo are linked to revivals of coherences via Eq.\,\eqref{eq:lossi} and hence indicate back-flow of information from environment to system.

It should be pointed out, that the upper and lower bounds of the Loschmidt echo $L_\phi$ could be made tighter by considering general $\alpha$-fidelities in the above Ineq.\,\eqref{eq:loschmidt}. However, since solving $L_\phi$ from Ineq.\,\eqref{eq:loschmidt} in such a general situation can only be done numerically, we have left this examination as a topic for future investigation.

\section{Conclusions and Outlook} \label{sec:conc}

In this paper we have introduced a novel family of $\alpha$-fidelities of quantum channels. We also have derived an inequality between the $\alpha$-fidelity of the initial states of the environment and the $\alpha$-fidelity of the corresponding induced dynamics: $F_\alpha(\xi_1,\xi_2) \leq \mc F_\alpha\big(\mc E_1^{(t)}, \mc E_2^{(t)}\big)$, $\alpha\in[1/2,1)$. From a more practical viewpoint, the inequality was then considered in the context of four different applications: {\bf A)} quantum programming, {\bf B)} discrimination of environmental properties, {\bf C)} quantum thermometry and in {\bf D)} deriving bounds for the Loschmidt echo of a general (pure) initial state of the spin-chain.

In a nutshell, the general procedure of estimating an unknown $\xi_2$
requires 1) the possibility to prepare the environment in a known ''calibration'' state $\xi_1$ and 2) the possibility to
compare the probe dynamic $\mc E_1^{(t)}$ and  $\mc E_2^{(t)}$ induced by these states. To apply our framework to quantum probing experimentally, as proposed in examples {\bf B)}, {\bf C)} and {\bf D)}, the experimenter needs first to perform quantum process tomography and then numerical optimization for all times, which may be extremely difficult. However, the starting point underlying behind our results is the Ineq.\,\eqref{eq:mainineq2}, which shows that non-trivial bounds for $F_{\alpha}(\xi_1, \xi_2)$ may be found
for plethora of choices of initial probe states $\varrho_1$ and $\varrho_2$; the notion of $\alpha$-fidelity of dynamics was developed as the best achievable resolution of this inequality. In particular, choosing $\varrho_i= \varrho,$ $i=1,2$, simplifies the Ineq.\,\eqref{eq:mainineq2} to
\begin{eqnarray}
F_{\alpha}(\xi_1, \xi_2) \leq F_{\alpha}\big(\mc E_1^{(t)} (\varrho) , \mc E_2^{(t)} (\varrho) \big) \leq 1 \, ,
\end{eqnarray}
for all $\varrho \in \mc S(\hi_S)$ and $\alpha \in [1/2, 1)$ and any given time $t\in [0,\infty)$. Although the resulting upper bound $F_{\alpha}\big(\mc E_1 (\varrho) , \mc E_2 (\varrho) \big)$ may not be the optimal one, any choice of the initial state of $\varrho \in \mc S(\hi_S)$ allows one to extract some information of system of interest via the protocols described above. In addition, the full process tomography is now reduced from the full channel tomography to mere state tomography of the evolved states of the probe $\mc E_1 (\varrho)$ and $\mc E_2 (\varrho)$.

%In addition to these applications, we believe that the inequality  $F_\alpha(\xi_1,\xi_2) \leq \mc F_\alpha\big(\mc E_1^{(t)}, \mc E_2^{(t)}\big)$, $\alpha\in[1/2,1)$, has significant potential to be utilized in also plethora of other quantum applications.

In our theoretical framework the introduced fidelities are based on the definition of the $\alpha$-R\'enyi divergence, given in Eq.\,\eqref{eq:defrenyi}, mainly because of its connection to the Uhlmann fidelity for $\alpha = 1/2$. However as we mentioned before, there exists another quantum divergence often encountered in the literature, namely $\widetilde{S}_{\alpha}$ defined in Eq.\eqref{eq:defrenyi0}. It is known that also $\widetilde{S}_\alpha$ satisfies both the properties (S1)-(S3) and the data processing inequality (S4) for all $\alpha \in (0,1)$ \cite{Muller2013}. Hence, also for $\widetilde{F}_\alpha(\varrho_1, \varrho_2 ) \doteq \tr{\varrho_1^\alpha \varrho_2^{1-\alpha}}$, we have $\widetilde{F}_\alpha(\varrho_1,\varrho_2) \widetilde{F}_\alpha(\xi_1,\xi_2) \leq \widetilde{F}_\alpha\big(\mc E_1^{(t)}(\varrho_1), \mc E_2^{(t)}(\varrho_2)\big)$, $\alpha\in(0,1)$. Furthermore, for commuting states $[\varrho_1, \varrho_2]=0$ the quantities $S_\alpha$ and $\widetilde{S}_\alpha$ coincide \cite{Muller2013}.  Although it is known that $\widetilde{F}_\alpha \leq F_\alpha$, since $ S_\alpha \leq \widetilde{S}_\alpha $ \cite{Wilde2014, Datta2014}, an interesting question to investigate in the future is which of the two definitions $\widetilde{F}_\alpha$ or $F_\alpha$ leads to a tighter inequality between the inducing states and induced dynamics -- or whether such optimality even holds in general. It is noteworthy, that the results shown in Figs.\,\ref{fig:hamiltonians}, \ref{fig:thermo} and \ref{fig:loschmidt} are independent of the choice of  $\widetilde{F}_\alpha$ or $F_\alpha$ due to the facts that the time-evolved final states of the probe are commuting in all the three cases as well as the environment states in applications {\bf B)}, {\bf C)}, and since $\widetilde{F}_\alpha$ and $F_\alpha$ coincide when comparing pure states in {\bf D)}.

Due to the generality of our formalism, our results can be used in different fields of physics: from quantum information theory to solid state physics, from particle physics to cosmology, from quantum gravity to atomic and molecular physics, and from quantum optics to quantum thermodynamics. Therefore our approach has the potential to pave the way to new fundamental theoretical and experimental discoveries in numerous physical scenarios.

\section*{Acknowledgements} 
The authors are grateful to M\'ario Ziman for his comments on the manuscript. 
S.M was supported by the EU Collaborative project QuProCS (Grant Agreement 641277) and the Academy of Finland (Project no. 287750).
M.T. and H.L. acknowledge financial support from the University of Turku
Graduate School (UTUGS) and M.T. additionally acknowledges the hospitality of Mathematical Physics Group in Toru\'n.
G.S. was partially supported by the National Science Centre (project 2015/17/B/ST2/02026), the Magnus Ehrnrooth Foundation and wishes to acknowledge the hospitality of Turku Quantum Technology Group.

%%%%%%%%%% APPENDICES %%%%%%%%%%%%%%%

\appendix

\section*{APPENDIX A: Evaluating $\alpha$-fidelity of qubit channels}
In this section we present the method that has been used to evaluate the $\alpha$-fidelity of one-qubit channels. First we simplify the analytic expression to make the numeric optimization more feasible. The density matrices of two arbitrary one-qubit states $\varrho_k$, $k=1,2$, can be written in the Bloch form as
\begin{eqnarray}
 \varrho_k &=& \frac 12 \left(\id + x_k \sigma_1 + y_k \sigma_2 + z_k \sigma_3 \right) \nonumber \\
 &=& \frac 12 \left( \begin{array}{cc} 1+z_k & x_k-iy_k \\ x_k+iy_k & 1-z_k \end{array} \right)\,,
\end{eqnarray}
where $x_k,\,y_k$ and $z_k$ are real numbers satisfying $x_k^2 + y_k^2 + z_k^2\leq 1$ and $\sigma_1,\,\sigma_2$ and $\sigma_3$ are the Pauli operators. The density matrices can be diagonalized as
\begin{eqnarray}\label{diag}
&& \varrho_k = \frac 1{2r_k(r_k+z_k)} \left( \begin{array}{cc} r_k+z_k & -(x_k-iy_k) \\ x_k+iy_k & r_k+z_k \end{array} \right) \times \nonumber \\ 
&& \left( \begin{array}{cc} \frac{1+r_k}2 & 0 \\ 0 & \frac{1-r_k}2 \end{array} \right)% \cdot
 \left( \begin{array}{cc} r_k+z_k & -(x_k-iy_k) \\ x_k+iy_k & r_k+z_k \end{array} \right)^\dagger,
\end{eqnarray}
where $r_k=\sqrt{x_k^2+y_k^2+z_k^2}$. Now by using Eq.~\eqref{diag} we can simplify the form of $\varrho_k^{\frac{1-\alpha}{\alpha}}$ as% are able to write the expression for the $\frac{1-\alpha}\alpha$ power of $\varrho_k$:
\begin{equation}\label{expo}
 \varrho_k^{\frac{1-\alpha}\alpha} = \frac 1{2r_k} \left( \begin{array}{cc} r_kA_k^++z_kA_k^- & (x_k-iy_k)A_k^- \\ (x_k+iy_k) A^- & r_kA_k^+-z_kA_k^- \end{array} \right),
\end{equation}
 where $A_k^{\pm}=\left(\frac{1+r_k}2\right)^{\frac{1-\alpha}\alpha}\pm\left(\frac{1-r_k}2\right)^{\frac{1-\alpha}\alpha}$. %, A_k^-=\left(\frac{1+r_k}2\right)^{\frac{1-\alpha}\alpha}-\left(\frac{1-r_k}2\right)^{\frac{1-\alpha}\alpha}$.
  By using Eq.~\eqref{expo} and the properties of trace and determinant, we can calculate the trace and determinant of the matrix product $\varrho_2^{\frac{1-\alpha}{2\alpha}}\varrho_1\varrho_2^{\frac{1-\alpha}{2\alpha}}$ as
 \begin{eqnarray}\label{trace}
  T &:=& \text{tr}\left[\varrho_2^{\frac{1-\alpha}{2\alpha}}\varrho_1\varrho_2^{\frac{1-\alpha}{2\alpha}}\right] = \frac 12 A^+_2 + \frac 12 A^-_2 \frac{r_1 r_2}{r_2}\,, \\ \label{det}
  D &:=& \det\left(\varrho_2^{\frac{1-\alpha}{2\alpha}}\varrho_1\varrho_2^{\frac{1-\alpha}{2\alpha}}\right) = \left( \frac{1-r_2^2}4 \right)^{\frac{1-\alpha}{\alpha}} \left( \frac{1-r_1^2}4 \right)\,. \qquad
 \end{eqnarray}
From Eqs.~\eqref{trace} and \eqref{det} we can write the eigenvalues $\lambda_+$ and $\lambda_-$ of $\varrho_2^{\frac{1-\alpha}{2\alpha}}\varrho_1\varrho_2^{\frac{1-\alpha}{2\alpha}}$ as
 \begin{equation}
  \lambda_\pm=\frac 12 \left( T \pm \sqrt{T^2 - 4D} \right).
 \end{equation}
Using these eigenvalues, we can diagonalize $\varrho_2^{\frac{1-\alpha}{2\alpha}}\varrho_1\varrho_2^{\frac{1-\alpha}{2\alpha}}$ as 
\begin{equation}
\varrho_2^{\frac{1-\alpha}{2\alpha}}\varrho_1\varrho_2^{\frac{1-\alpha}{2\alpha}} = URU^{\dagger}\,,
\end{equation}
where $R =$ diag($\lambda_+, \lambda_-$) and $U$ is the diagonalizing unitary matrix found in the same way as in Eq.~\eqref{diag}. Finally, this allows us to express the $\alpha$-fidelity of two arbitrary one-qubit states, $\varrho_1$ and $\varrho_2$, as
 \begin{equation}
  F_\alpha(\varrho_1,\varrho_2) = \text{tr}\left[\left( URU^{\dagger} \right)^{\alpha}\right] = \lambda_+^\alpha + \lambda_-^\alpha\,.
 \end{equation}
Now we see that, in order to calculate the channel $\alpha$-fidelity of one-qubit channels $\mc E_1$ and $\mc E_2$, it is enough to minimize the quotient
 \begin{equation}\label{fraction}
\frac{F_{\alpha} \big( \mc E_1(\varrho_1), \mc E_2(\varrho_2)\big) }{F_{\alpha} (\varrho_1,\varrho_2)} =  \frac{\widetilde{\lambda}^\alpha_+ + \widetilde{\lambda}_-^\alpha}{\lambda_+^\alpha + \lambda_-^\alpha}\,.  
 \end{equation}
 Here $\widetilde{\lambda}_\pm$ are eigenvalues of $\mc E_2(\varrho_2)^{\frac{1-\alpha}{2\alpha}}\mc E_1(\varrho_1)\mc E_2(\varrho_2)^{\frac{1-\alpha}{2\alpha}}$ calculated in a similar way from the Bloch coordinates of $\mc E_1(\varrho_1)$ and $\mc E_2(\varrho_2)$, and $\lambda_\pm$ are the above solved eigenvalues of $\varrho_2^{\frac{1-\alpha}{2\alpha}}\varrho_1\varrho_2^{\frac{1-\alpha}{2\alpha}}$, where $\varrho_1$ and $\varrho_2$ are the initial states of the qubits. The quantity to be minimized in Eq.~\eqref{fraction} is a function of six variables $x_1,y_1,z_1,x_2,y_2,z_2$ under the constraints $x_1^2+y_1^2+z_1^2\le 1$ and $x_2^2+y_2^2+z_2^2\le 1$. 
  
 \subsection*{Example: the case of pure dephasing channels}
  
 In this example the reduced dynamics is a family of dephasing channels with the same invariance basis for all values of time $t$ and the inverse temperature $\beta=1/k_B T$. Let $\Gamma_1(t)$ and $\Gamma_2(t)$ be the decoherence rates at time $t$ (the same for both channels) corresponding to $\beta_1$ and $\beta_2$, respectively. The two dephasing channels $\mc E_i$ are defined via 
\begin{eqnarray}
&&\varrho=\frac 12 \left( \begin{array}{cc} 1+z & x-iy \\ x+iy & 1-z \end{array} \right) \nonumber \\ &\mapsto& \mc E_i^{(t)}(\varrho)= \frac 12 \left( \begin{array}{cc} 1+z & \Gamma_i(t) (x-iy) \\\Gamma_i(t)^* (x+iy) & 1-z \end{array} \right)\,.\qquad
\end{eqnarray}
  We have performed the numerical minimization of Eq.~\eqref{fraction} by tools of constrained minimisation. More precisely, we have used the method SLSQP (Sequential Least Square Programming) implemented in \texttt{scipy.optimize} library.
 %The minimum is attained for the pair of identical pure states mutually unbiased to the pointer basis of decoherence and is equal
  The numerical calculation shows, that the minimum is attained when $\varrho_1=\varrho_2$ is any pure state maximally unbiased with respect to the invariant basis, for example $|+\rangle\langle + | = \frac 12 (\id + \sigma_1)$. This observation lets us write the explicit formula for $\alpha$-fidelity of the dynamics: 
\begin{eqnarray} \label{dec_fid}
 \mathcal{F}_\alpha(\mc E_1^{(t)},\mc E_2^{(t)}) &=& \left(\frac{1+\Gamma_2(t)}2\right)^{1-\alpha}  \left(\frac{1+\Gamma_1(t)}2\right)^\alpha \nonumber \\
 &+& \left(\frac{1-\Gamma_2(t)}2\right)^{1-\alpha} \left(\frac{1-\Gamma_1(t)}2\right)^\alpha.\qquad
\end{eqnarray}

\section*{APPENDIX B}\label{app:b}

In application {\bf B)} (Ruling out incompatible environmental properties) we consider a system coupled to a thermal environment with Hamiltonian $\Ho_E$ in two different inverse temperatures $\beta_i =1/T_i$, $i=1,2$; we have used natural units $\hbar=1=k_B$ throughout this section. We showed that the partition functions $Z(\beta_i, \Ho_E)$ of the two environmental states and the corresponding induced dynamics $\mc E_i^{(t)}$ satisfy the inequality Ineq.\,\eqref{eq:mainineq3}. It was reasoned that some of the properties of the environment inscribed in $\Ho_E$ may be ruled out by measuring the dynamics induced by two different temperatures and by testing which partition functions are (in)compatible with Ineq.\,\eqref{eq:mainineq3}. In a particular case of environment consisting of harmonic oscillators, that is $\Ho_E = \sum_k \omega_k \big(b_k^{\dagger} b_k + \frac{1}{2}\id\big)$, where $b^{\dagger}_k$ and $b_k$ are the creation and annihilation operators of the environmental mode $k$, respectively, the logarithms of the partition functions attain a simple form
\begin{eqnarray}\label{eq:decopartition}
\ln \left[ Z(\beta_i, \Ho_E) \right] = \sum_k \, \big( - \frac{\beta_i \omega_k}{2} -\ln\left[1-e^{-\beta_i \omega_k} \right] \big)\,, \qquad
\end{eqnarray}
$i=1,2,$ which then only depends on the oscillators frequencies $\omega_k$ and the known inverse temperatures $\beta_i$. 

Inequality \,\eqref{eq:mainineq3} holds regardless of the choice of the actual coupling between the system and the environment. As an example of our method, in Fig.\,\ref{fig:hamiltonians} we have, however, deployed a particular exactly solvable system-environment model leading to the dephasing channel of the system described in Ref.\,\cite{Breuer}. In this model the system is a qubit coupled to environment consisting of harmonic oscillators and the total Hamiltonian, consisting of the Hamiltonians of the system $\Ho_S$, the environment $\Ho_E$, and the interaction Hamiltonian $\Ho_I$, governing the composite evolution can be written as 
\begin{align}\nonumber
\Ho	&= \Ho_S + \Ho_E + \Ho_I\\
	&=\frac{\omega_0}{2}\sigma_3 + \sum_k \omega_k \big(b_k^{\dagger} b_k + \frac{1}{2}\id\big) + \sum_k \sigma_3\otimes \big(g_k b_k^{\dagger} + g_k^* b_k\big)\,,
\end{align}
$g_k$ describes the strength at which the system couples to different modes. Fixing the initial composite system state as $\varrho_{S + E} = \varrho \otimes \xi(\beta_i)$ we get an analytical solution for the qubit system dynamics
\begin{equation}\label{decohchannel}
\varrho = 
\begin{pmatrix}
\varrho_{00}	& \varrho_{01}\\
\varrho_{10}	& \varrho_{11}
\end{pmatrix}
\mapsto \mc E_i^{(t)} (\varrho) =
\begin{pmatrix}
\varrho_{00}	& \Gamma_i(t) \, \varrho_{01}\\
\Gamma_i(t)\, \varrho_{10}	& \varrho_{11}
\end{pmatrix}\,,
\end{equation}
where $\Gamma_i (t) = \text{exp}\big[-\sum_k \frac{4 |g_k|^2}{\omega_k^2}\coth\big(\frac{\omega_k}{2T_i}\big)\big[1-\cos (\omega_k t)\big]\big]$; we refer the reader to \cite{Breuer} for further details of this model. 

For simplicity we will only consider environment consisting of a single oscillator; this is also the case in Fig.\,\ref{fig:hamiltonians}. By fixing two temperatures of the environmental initial state, $T_1,\, T_2$, and a single oscillator frequency $\omega$ which we want to probe, we obtain two dephasing channels $\mc E_1^{(t)}$ and $\mc E_2^{(t)}$ of the form Eq.~\eqref{decohchannel}
%----------------------------------------
and according to (\ref{dec_fid}) their $\alpha$-fidelity is equal to
\begin{eqnarray}
&& \mc F_{\alpha} (\mc E^{(t)}_1, \mc E^{(t)}_2)= \nonumber \\
&&\left(\frac{1+C_g(t)^{\coth(\frac 12 \omega \beta_2)}}2\right)^{1-\alpha} \left(\frac{1+C_g(t)^{\coth(\frac 12 \omega \beta_1)}}2\right)^\alpha \nonumber \\
&+& \left(\frac{1-C_g(t)^{\coth(\frac 12 \omega \beta_2)}}2\right)^{1-\alpha} \left(\frac{1-C_g(t)^{\coth(\frac 12 \omega \beta_1)}}2\right)^\alpha ,\nonumber \\
\end{eqnarray}
for $C_g(t)=\exp\big[-\frac{4 |g|^2}{\omega^2}\big[1-\cos{\omega t}\big]\big]$.
%------------------------------------------
%As mentioned above, the $\alpha$-fidelity of these channels when choosing $\varrho_1 = \varrho_2 = \ket{+}\bra{+}$;
%we point out that this minimization is independent of the value of the coupling strength $g$.
Importantly, the function $C_g(t)$ is the only term involving both the coupling strength $g$ and time $t$. It is easily verified that for any $|g|/\omega \geq 1/\sqrt 8$ the function $C_g(t)$ can be kept constant in value by choosing $t$ appropriately. Therefore, we can conclude that any coupling (in units of $\omega$) stronger than $1/\sqrt 8$ can only make the number $\inf_t\mc F_{\alpha} (\mc E^{(t)}_1, \mc E^{(t)}_2)$ smaller and consequently our protocol to work better. For demonstration we have chosen $|g|/\omega = 1$, $T_1/\omega=0.25$ and $T_2/\omega = 0.75$. We can therefore plot $\inf_t \ln\big[\mc F_{\alpha} (\mc E^{(t)}_1, \mc E^{(t)}_2)\big]$ as a function of $\alpha$, which is the r.h.s.~of  Ineq.\,\eqref{eq:mainineq3}. Also, by using the same temperatures $T_1,\, T_2$ and varying the frequency value in Eq.~\eqref{eq:decopartition} we can plot the l.h.s.~of Ineq.\,\eqref{eq:mainineq3} as a function of $\alpha$ and compare it with the plot of the r.h.s. Whenever we see violation of Ineq.\,\eqref{eq:mainineq3}, we know that these frequency values cannot have lead to the induced dynamics. In accordance, plots with frequencies satisfying and violating Ineq.\,\eqref{eq:mainineq3} are presented in Fig.\,\ref{fig:hamiltonians}. 

%We wish to emphasize, that in order perform the protocol, we do not need to know what is the analytical expression of dynamics or even the coupling between the system and environment. This is due to the fact that in experiment we can prepare the probe system and let it interact with its environment. By performing process tomography, we obtain the dynamics of the state, which lets us plot the r.h.s.~of \color{blue}Ineq.~\eqref{eq:discr}\color{black} as a function of $\alpha$. After that we can just test which environmental state parameters lead to violation of \color{blue}Ineq.~\eqref{eq:discr}\color{black}, when inserted in the l.h.s.

\section*{APPENDIX C}\label{app:c}
In application {\bf C)} (Quantum thermometry) we continue on considering a system coupled to a thermal environment. This time we assume to know the oscillator frequencies of the Hamiltonian $\Ho_E = \sum_k \omega_k \big(b_k^{\dagger} b_k + \frac{1}{2}\id\big)$ and again we check the induced dynamics for two different temperatures: $T_0 = 0$ and $T$ that is unknown. Since $\xi(\infty)=\bigotimes_k |0\rangle \langle 0|$, the $\alpha$-R\'enyi divergence between the initial states of the environment in the inverse temperatures $\beta_0 = \infty$ and $\beta=1/k_B T$ can be calculated that
\begin{eqnarray}
&&S_\alpha\big(\xi(\beta_0), \xi(\beta)\big)\nonumber \\
&=&\frac{1}{\alpha -1} \ln\left[ \tr{\left( \xi(\beta)^{\frac{1-\alpha}{2\alpha}} \, \xi(\beta_0) \, \xi(\beta)^{\frac{1-\alpha}{2\alpha}} \right)^\alpha } \right] \nonumber \\
&=&\frac{1-\alpha}{\alpha -1} \ln\left[ e^{-\frac{1}{2} \sum_k  \beta \hbar \omega_k} / Z(\beta, \Ho_E)  \right] \nonumber \\
&=& \frac{1-\alpha}{\alpha -1} \ln\left[ \prod_k \big( 1-e^{-\beta \hbar \omega_k} \big) \right]\, ,
\end{eqnarray}
with $\beta=1/k_B T$ and the last equality follows from inserting $Z(\beta, \Ho_E)$ from Eq.\,\eqref{eq:decopartition}. Similarly, we confirm that
\begin{eqnarray}
S_\alpha\big(\xi(\beta),\xi(\beta_0)\big) &=& \frac{\alpha}{\alpha -1} \ln\left[\prod_k \big(1-e^{-\beta \omega_k}\big) \right]\, .
\end{eqnarray}
Consequently, from Ineq.\,\eqref{eq:mainineq1} one confirms that
\begin{eqnarray}\label{eq:thermo}
&& \prod_k \, \big( 1-e^{-\hbar \omega_k /k_B T}\big) \nonumber \\
&\leq& \inf_{t\geq 0} \left\{ \begin{array}{lc}
\mc F_{\alpha} \big(\mc E_T^{(t)},\mc E_0^{(t)}\big)^{\frac{1}{\alpha}} , & \text{for} \  \alpha \in (0, \frac{1}{2}) \\
\mc F_{\alpha} \big(\mc E_0^{(t)},\mc E_T^{(t)}\big)^{\frac{1}{1-\alpha}}, & \text{for} \ \alpha \in [\frac{1}{2}, 1)
\end{array} \right. ,
\end{eqnarray}
where $\mc E_0^{(t)}$ and $\mc E_T^{(t)}$ are the dynamics induced by $\beta_0$ and $\beta_T$, respectively. From this inequality the temperature $T$ can be numerically estimated. For simplicity, we again consider the case of only a single harmonic oscillator mode. In this case the above relation can be solved in terms of the temperature $T$ analytically: the solution is presented in Ineq.\,\eqref{eq:thermoineq1}. 

Notice, that the above Ineq.\,\eqref{eq:thermo} holds model-independently even without {\it a priori} knowledge of the coupling between the system and the environment. As a concrete example, in Fig.\,\ref{fig:thermo} we have demonstrated the power of our method by considering the Jaynes-Cummings -model of a qubit coupled to a thermal environment consisting of a single harmonic oscillator. In such a model the total Hamiltonian of the system-environment composite reads
\begin{eqnarray}
\Ho &=& \Ho_S + \Ho_E + \Ho_I \nonumber \\
	&=& \omega_0 \sigma_+ \sigma_-  + \omega \big( b^\dagger b + \frac{1}{2} \id \big) + g \big( \sigma_+ \otimes b + \sigma_-\otimes b^\dagger \big)\,,\qquad
\end{eqnarray}
where $\sigma_+ = |1 \rangle \langle 0|$ and $\sigma_- = |0\rangle\langle 1|$ are the raising and
lowering operators of the qubit system, respectively, and $g$ describes the coupling strength. Assume that an initial separable state of the composite system $\varrho \otimes \xi(\beta)$ evolves according to this Hamiltonian. Then, in the resonant case of matching system and environment frequencies $\Delta = \omega - \omega_0 =0$, the reduced dynamics take a relatively simple form
\begin{eqnarray}
&& \varrho = 
\left(\begin{array}{lc}
\varrho_{00} & \varrho_{01}\\
\varrho_{10} & \varrho_{11}
\end{array}\right)
\mapsto \mc E_T^{(t)} (\varrho) = \nonumber \\ 
&& \left(\begin{array}{lc}
a_T(t)\,\varrho_{00} + \big[ 1-b_T(t) \big]\,\varrho_{11}	& 	c_T(t)^* \, \varrho_{01}  \\
c_T(t)\, \varrho_{10}	& \big[ 1-a_T(t) \big]\, \varrho_{00} + b_T(t)\, \varrho_{11}
\end{array}\right),\nonumber \\ 
\end{eqnarray}
where the coefficients
\begin{eqnarray}
a_T(t) &=& \tr {\Co^\dagger (\hat n, t) \Co (\hat n, t) \, \xi(\beta)}\, , \nonumber \\
b_T(t) &=& \tr {\Co^\dagger (\hat n+1, t) \Co(\hat n+1, t) \, \xi(\beta)}\, ,\nonumber \\
c_T(t)&=& \tr {\Co (\hat n+1, t) \Co (\hat n,t) \, \xi(\beta)}
\end{eqnarray}
depend on the operator $\Co(\hat n,t) = \cos\big( \abs{g} t \, \sqrt{\hat n} \big)$, $\hat n = b^\dagger b$ \cite{Smirne2010}. Inserting the thermal state $\xi(\beta)$ in the above formulas we get, e.g., $a_T(t) =\sum_{n = 0}^N e^{-\beta\hbar \omega (n + 1/2)} \cos^2\big(|g| t \sqrt n \big) / \sum_{n = 0}^\infty e^{-\beta\hbar\omega(n+1/2)}$. Although in principle $N\rightarrow\infty$, in lack of a closed form solution we have truncated the series to $N=10$ (similarly for $b_T(t)$ and $c_T(t)$). The error due to this estimation is negligible because for large $N$ the dominating exponent factor makes the summand minuscule. 

Our numerics suggest that, even though the induced dynamics are not purely dephasing, the choice $\varrho_1=\varrho_2=|+\rangle\langle +|$ minimizes the channel $\alpha$-fidelity $\inf_{t\geq 0}\mc F_\alpha \big( \mc E_0^{(t)}, \mc E_T^{(t)}\big)$ also in this case. Therefore, the temperature can, in fact, be directly estimated in terms of the Kullback-Leibler divergence $S_1\big(\mc E_0(\varrho)||\mc E_T(\varrho) \big)$, $\varrho = |+\rangle\langle +|$, as shown in Ineqs.\,\eqref{eq:thermoineq3} and \eqref{eq:thermoineq4} and which we have plotted in Fig.\,\ref{fig:thermo}. The remaining optimization with respect to time $t$ has been done numerically by using the build in algorithms of Mathematica. Finally, we note that, because of the relation of the coupling strength $g$ and $t$ in the terms $\Co(\hat n,t)$ defining the dynamics and the time-optimization involved in our protocol, the upper and lower bounds presented in Fig.\,\ref{fig:thermo} are independent of the strength of the coupling as long as it is non-vanishing.


\begin{thebibliography}{10}
%1
\bibitem{Breuer}
H.P. Breuer and F. Petruccione,
\newblock {\it The Theory of Open Quantum Systems}
\newblock (Oxford University Press, New York, 2007), Vol. 1, p. 227 -- 228.
%2
\bibitem{recati2005}
A. Recati, P. O. Fedichev, W. Zwerger, J. von Delft, and P. Zoller,
\newblock {Phys. Rev. Lett.} {\bf 94}, 040404 (2005).
%3
\bibitem{bruderer2006}
M. Bruderer and D. Jaksch,
\newblock {New J. Phys.} {\bf 8}, 87 (2006).
%4
\bibitem{johnson2011}
T. H. Johnson, S. R. Clark, M. Bruderer and D. Jaksch,
\newblock {Phys. Rev. A} {\bf 84}, 023617 (2011).
%5
\bibitem{dorner2013}
R. Dorner, S. R. Clark, L. Heaney, R. Fazio, J. Goold and V. Vedral,
\newblock {Phys. Rev. Lett.} {\bf 110}, 230601 (2013).
%6
\bibitem{mcendoo2013}
S. McEndoo, P. Haikka, G. De Chiara, G. M. Palma and S. Maniscalco,
\newblock {Europhys. Lett.} {\bf 101}, 60005 (2013).
%7
\bibitem{haikka2013}
P. Haikka, S. McEndoo and S. Maniscalco,
\newblock {Phys. Rev. A} {\bf 87}, 012127 (2013).
%8
\bibitem{haikka2014}
P. Haikka and K. M{{\o{}}}ølmer,
\newblock {Phys. Rev. A} {\bf 89}, 052114 (2014).
%9
\bibitem{Anil151}
V. Jagadish and A. Shaji.
%\newblock The dynamics of a qubit reveals its coupling to a N level system.
\newblock Ann. Phys. {\bf 362}, pp. 287-297 (2015).
%10
\bibitem{Anil152}
S. Manikandan, V. Jagadish and A. Shaji.
%\newblock{The qubit reveals a qubit-cavity system}
\newblock arXiv:1507.00583
%%%%%%%%%%%%%%%%%%%%%%%%%%%%%%%%%%%%
%11
\bibitem{NielsenChuang}
M.A. Nielsen and I.L. Chuang,
\newblock {Phys. Rev. Lett.} {\bf 79}: 321--324 (1997).
%12
\bibitem{Hillery02}	M. Hillery, M. Ziman and V. Bu\v{z}ek. 
%\newblock {\it Implementation of quantum maps by programmable quantum processors},
\newblock {Phys. Rev. A} {\bf 66}, 042302 (2002).
%13
\bibitem{Hillery06}	M. Hillery, M. Ziman and V. Bu\v{z}ek. 
%\newblock {\it Approximate programmable quantum processors}, 
\newblock {Phys. Rev. A} {\bf 73}, 022345 (2006).
%14
\bibitem{HeiTuk2015}
T. Heinosaari and M. Tukiainen,
%\newblock {\it Notes on deterministic programming of quantum observables and channels},
\newblock {Quantum Inf. Processing}, 14:3097-3114 (2015).
%15
\bibitem{VidalarX}
G. Vidal and J. I. Cirac. 
%\newblock {\it Storage of quantum dynamics in quantum states: a quasi-perfect programmable quantum gate}, 
arXiv:quant-ph/0012067v1.
%16
\bibitem{Vidal02}
G. Vidal, L. Masanes and J. I. Cirac.
%\newblock {\it Storing Quantum Dynamics in Quantum States: A Stochastic Programmable Gate},
\newblock{Phys. Rev. Lett.} {\bf 88}, (2002).
%17
\bibitem{Hillery01}
M. Hillery, V. Bu\v{z}ek and M. Ziman.
%\newblock {\it Probabilistic implementation of universal quantum processors},
\newblock{Phys. Rev. A} {\bf 65}, 022301 (2002).
%18
\bibitem{Hillery04}	
M. Hillery, M. Ziman and V. Bu\v{z}ek. 
%\newblock {\it Improving performance of probabilistic programmable quantum processors},
\newblock{Phys. Rev. A} {\bf 69}, 042311 (2004).

%19
\bibitem{Renyi}
A. R\'enyi.
\newblock On measures of information and entropy. In: Proceedings of the Fourth Berkeley Symposium
on Mathematical Statistics and Probability, vol. 1: Contributions to the Theory of Statistics, pp. 547–
561. 
\newblock (University of California Press, Oakland, 1961 )

\bibitem{Muller2013}
M. M\"uller-Lennert, F. Dupuis, O. Szehr, S. Fehr and M. Tomamichel,
%\newblock {\it On quantum R\'enyi entropies: A new generalization and some properties},
\newblock {J. Math. Phys.} {\bf 54}, 122203 (2013).
%20
\bibitem{Wilde2014}
M. Wilde, A. Winter and D. Yang,
%\newblock {\it Strong Converse for the Classical Capacity of Entanglement-Breaking and Hadamard Channels via a Sandwiched R\'eényi Relative Entropy.
\newblock {Comm. Math. Phys.} {\bf 331}, pp. 593-622 (2014).
%21
\bibitem{Beigi2013}
 S. Beigi,
%\newblock {\it Sandwiched R\'enyi divergence satisfies data processing inequality},
\newblock {J. Math. Phys.} {\bf 54}, 122202 (2013).
%22
\bibitem{Datta2014}
 N. Datta and F. Leditzky,
%\newblock {\it A limit of the quantum R\'eényi divergence},
\newblock {J. Phys. A: Math. Theor.} {\bf 46}, 045304 (2014).
%23
\bibitem{Datta2016}
 K. Audenaert and N. Datta,
%\newblock {\it $\alpha-z$-R\'enyi relative entropies},
\newblock {\it arXiv}:1310.7178v3.
%24
\bibitem{Carlen2016}
 E. Carlen, R. Frank and E. Lieb,
%\newblock {\it Some operator and trace function convexity theorems},
\newblock {Linear Algebra and its Applications} {\bf 490}, 174–185 (2016).
%25
\bibitem{Raginsky2001}
M. Raginsky. 
%\newblock {\it A fidelity measure for quantum channels},
\newblock{Phys. Lett. A} {\bf 290}, 11–18 (2001).
%26
\bibitem{Raginsky2005}
V. Belavkin, G. D'Ariano and M. Raginsky.
%\newblock {\it Operational distance and fidelity for quantum channels},
\newblock{J. Math. Phys.} {\bf 46}, 062106 (2005).
%27
\bibitem{Mario2008}
M. Ziman.
%\newblock {\it Process positive-operator-valued measure: A mathematical framework for the description of process tomography experiments},
\newblock{Phys. Rev. A} {\bf 77}, 062112 (2008).
%28
\bibitem{Mario2009}
M. Ziman and M. Sedl\'{a}k.
%\newblock {\it Single-shot discrimination of quantum unitary processes},
\newblock{J. Mod. Opt.} {\bf 57}, 253–259 (2010).
%29
\bibitem{Nielsen2005}
A. Gilchrist, N. Langford and M. Nielsen,
%\newblock Distance measures to compare real and ideal quantum processes.
\newblock {Phys. Rev. A} {\bf 71}, 062310 (2005).
%30
\bibitem{NielsenInfo}
M. Nielsen and I. Chuang,
\newblock {\it Quantum Computation and Quantum Information -- 10th Anniversary Edition}
\newblock (Cambridge University Press, New York, 2010).
%31
\bibitem{Mosonyi2014}
M. Mosonyi,
%\newblock Inequalities for the quantum R\'eényi divergences with applications to compound coding problems.
\newblock {\it arXiv}:1310.7525v3.
%32
\bibitem{Buzek2006}
V. Bu\v{z}ek, M. Hillery, M. Ziman and M. Ro\v{s}ko,
%\newblock Programmable Quantum Processors.
\newblock {Quantum Inf. Process.} {\bf 5}, 313 (2006).
%33
\bibitem{Ruskai2002}
M.B. Ruskai, S. Szarek, and E. Werner.
%\newblock An analysis of completely-positive trace-preserving maps on 2x2 matrices.
\newblock Lin. Alg. Appl. {\bf 347}:159–187, (2002).


\bibitem{Miszczak09} J. Miszczak et. al.
%\newblock SUB-AND SUPER-FIDELITY AS BOUNDS FOR QUANTUM FIDELITY.
\newblock {Quantum Inf. Comput.} {\bf 9}, 0103 (2009).
%34
\bibitem{Lieb1976}
E. Lieb and W. Thirring,
\newblock {\it Studies in mathematical physics}.
\newblock (Princeton University Press, Princeton, 1976).
%35
\bibitem{Araki1990}
H. Araki,
%\newblock On an inequality of Lieb and Thirring.
\newblock {Lett. Math. Phys.} {\bf 19}, pp. 167-170, (1990).
%???
\bibitem{Smirne2010}
A. Smirne and B. Vacchini,
%\newblock Nakajima-Zwanzig versus time-convolutionless master equation for the non-Markovian dynamics of a two-level system.
\newblock {Phys. Rev. A} {\bf 82}, 022110 (2010).
%%%%%%%%%%%%%%%%%%%%%%%
%%%%%%%%%%%%%%%%%%%%%%%
\bibitem{Jevtic2015}
S. Jevtic, D. Newman, T. Rudolph, and T. M. Stace,
%\newblock Qubit-assisted thermometry of a quantum harmonic oscillator.
\newblock {Phys. Rev. A} {\bf 91}, 012331 (2015).
%%%%%%%%%%%%%%%%%%%%%%%
%%%%%%%%%%%%%%%%%%%%%%%
\bibitem{Brunelli2012}
M. Brunelli, S. Olivares, M. Paternostro and M. G. A. Paris,
%\newblock Qubit-assisted thermometry of a quantum harmonic oscillator.
\newblock {Phys. Rev. A} {\bf 86}, 012125 (2012).
%
\bibitem{Higgins2013}
K. D. B. Higgins, B. W. Lovett and E. M. Gauger,
%\newblock Quantum thermometry using the ac Stark shift within the Rabi model.
\newblock {Phys. Rev. B} {\bf 88}, 155409 (2013).
%
\bibitem{Johnson2016}
T. H. Johnson, F. Cosco, M. T. Mitchison, D. Jaksch and S. R. Clark,
%\newblock Thermometry of ultracold atoms via nonequilibrium work distributions.
\newblock {Phys. Rev. A} {\bf 93}, 053619 (2016).
%
\bibitem{Correa2015}
L. A. Correa, M. Mehboudi, G. Adesso, and A. Sanpera,
%\newblock Individual Quantum Probes for Optimal Thermometry.
\newblock {Phys. Rev. Lett. } {\bf 114}, 220405 (2015).
%%%%%%%%%%%%%%%%%%%%%%%
%%%%%%%%%%%%%%%%%%%%%%%
%36
\bibitem{Rivas2014}
\'{A}. Rivas, S. Huelga and M. Plenio,
\newblock {Rev. Prog. Phys.} {\bf 77}, 094001 (2014).
%37
\bibitem{Breuer2016}
H.-P. Breuer, E.-M. Laine, J. Piilo, and B. Vacchini,
\newblock {Rev. Mod. Phys.} {\bf 88}, 021002 (2016).
%38
\bibitem{Hinarejos2016}
M. Hinarejos, M.-C. Ba\~nuls, A. P\'{e}rez and I. de Vega,
\newblock {arXiv:quant-ph/1606.01185v1}.
%39
\bibitem{Quan2006}
H.T. Guan, Z. Song, X.F. Liu, P. Zanardi and C.P. Sun,
%\newblock Decay of Loschmidt Echo Enhanced by Quantum Criticality.
\newblock{Phys. Rev. Lett.} {\bf 96}, 140604 (2006).
%40
\bibitem{Haikka2012}
P. Haikka, J. Goold, S. McEndoo, F. Plastina and S. Maniscalco,
%\newblock Non-Markovianity, Loschmidt echo, and criticality: A unified picture.
\newblock {Phys. Rev. A} {\bf 85}, 060101(R) (2012).



\end{thebibliography}
\end{document}